\documentclass[11pt,reqno]{amsart}
\usepackage[T1]{fontenc}
\usepackage{lmodern}
\usepackage[margin=1.12in]{geometry}
\usepackage{amsmath,amssymb,mathtools,mathrsfs}
\usepackage{booktabs,array,enumitem,longtable}
\usepackage{microtype}
\usepackage{xurl}
\usepackage{tikz}
\usetikzlibrary{arrows.meta,calc}
\usepackage{placeins}
\usepackage[foot]{amsaddr}

\usepackage[hidelinks,bookmarksnumbered]{hyperref}
\hypersetup{pdftitle={An Explicit Counterexample to Tsirelson's Problem via a Linear System Game},
  pdfsubject={Finite presentations, approximate representations, and non-local games}}
\newtheorem{theorem}{Theorem}[section]
\newtheorem{proposition}[theorem]{Proposition}
\newtheorem{lemma}[theorem]{Lemma}
\newtheorem{corollary}[theorem]{Corollary}
\theoremstyle{definition}
\newtheorem{definition}[theorem]{Definition}
\theoremstyle{remark}
\newtheorem{remark}[theorem]{Remark}
\numberwithin{equation}{section}
\newcommand{\F}{\mathbb F}
\newcommand{\Z}{\mathbb Z}
\newcommand{\C}{\mathbb C}
\newcommand{\N}{\mathbb N}
\newcommand{\R}{\mathbb R}

\newcommand{\norm}[1]{\left\lVert #1\right\rVert}

\newcommand{\M}{\mathcal M}
\newcommand{\HH}{\mathcal H}
\DeclareMathOperator{\EL}{EL}
\DeclareMathOperator{\SL}{SL}
\DeclareMathOperator{\GL}{GL}
\DeclareMathOperator{\Mat}{M}
\DeclareMathOperator{\U}{U}
\DeclareMathOperator{\tr}{tr}
\DeclareMathOperator{\Tr}{Tr}
\DeclareMathOperator{\defect}{err}

\newcommand{\lean}[1]{{\small\nolinkurl{#1}}}
\setlist[enumerate]{label=\textup{(\roman*)},leftmargin=2.2em,itemsep=3pt}
\setlist[itemize]{leftmargin=1.8em,itemsep=3pt}
\allowdisplaybreaks[1]
\title[An Explicit Counterexample to Tsirelson's Problem via a Linear System Game]{An Explicit Counterexample to Tsirelson's Problem via a Linear System Game}

\author{Minbo Gao$^{*}$}
\address{$^{*}$Institute of Software, Chinese Academy of Sciences, Beijing, China; University of Chinese Academy of Sciences, Beijing, China; Tencent Hunyuan, Beijing, China.}
\email{gmb17@tsinghua.org.cn}

\author{Tianshi Yu$^{\dagger}$}  
\address{$^{\dagger}$Institute of Software, Chinese Academy of Sciences, Beijing, China.} 
\email{yuts@ios.ac.cn}

\author{Lihong Zhi$^{\ddagger}$}
\address{$^{\ddagger}$State Key Laboratory of Mathematical Sciences, Academy of Mathematics and Systems Science, Chinese Academy of Science, Beijing, China; University of Chinese Academy of Science, Beijing, China.} 
\email{lzhi@mmrc.iss.ac.cn}
\date{}
\begin{document}
\begin{abstract}
We construct an explicit binary linear system game that
separates \(C_{qa}\), the closure of the set of finite dimensional quantum
correlations, from \(C_{qc}\), the set of commuting operator
correlations.
The game admits a perfect commuting operator strategy, while every
correlation in \(C_{qa}\) has a success probability strictly less than 1.
This provides a concrete counterexample to Tsirelson's problem
in its approximation form.
The defining system has \(1{,}417{,}152\) equations in
\(1{,}889{,}684\) variables, with exactly three nonzero coefficients
per equation and a single nonzero entry on the right hand side.
We also compute the classical value exactly.  The complete system is specified in Lean~4, and the separation
and exact classical value are formalized using Mathlib.
\end{abstract}
\maketitle

\section{Introduction}\label{sec:introduction}

Tsirelson's approximation problem asks whether correlations obtained
from commuting measurements can always be approximated by correlations
from measurements on finite dimensional tensor products
\cite{Tsirelson2006,ScholzWerner2008}.
Write \(C_q\) for the finite dimensional quantum correlations,
\(C_{qa}=\overline{C_q}\) for their closure, and \(C_{qc}\) for the
commuting correlations. The question is whether \(C_{qa}=C_{qc}\)
for every finite choice of questions and answers.
It is equivalent to Connes' embedding conjecture
\cite{Connes1976,JungeEtAl2011,Fritz2012,Ozawa2013}.
Ji, Natarajan, Vidick, Wright, and Yuen gave a negative answer using
recursive compression of nonlocal games~\cite{JiEtAl2020}.
Here we realize this separation by a concrete binary linear system
game whose complete defining data can be generated and stored.
We compute its classical value exactly, prove an explicit positive
lower bound for the gap between its commuting and quantum values,
and formalize the construction and these results in Lean~4.

Binary linear system games connect parity constraints with unitary
representations of finitely presented groups. Cleve and Mittal
characterized perfect finite dimensional strategies by operator
solutions~\cite{CleveMittal2014}. Cleve, Liu, and Slofstra associated
a solution group to each binary system and characterized perfect
commuting strategies by the nontriviality of a distinguished central
involution~\cite{CleveLiuSlofstra2017}.
Slofstra's embedding theorem transfers finitely presented group
constructions to this setting~\cite{Slofstra}.
His subsequent work produced a game with a perfect correlation in
\(C_{qa}\) but no perfect strategy on any tensor product of Hilbert
spaces~\cite{Slofstra2019}.
Thus failure to attain a perfect quantum strategy does not imply
that the quantum value is less than one. Our construction gives
a uniform bound below one, even as the dimensions grow.
The criterion of Slofstra and Vidick~\cite{SV} expresses this
requirement through approximate unitary representations of the
solution group.

The group construction builds on Kun and Thom's work on compressed
subgroups of Kazhdan groups~\cite{KunThom2026}.
Thom's normalization theorem~\cite{Thom}, made unconditional by
the spectral gap theorem of Alekseev, Liu, and Thom~\cite{ALT},
provides the obstruction to matrix approximation.
We use finitely presented Kazhdan covers~\cite{EJZ} to implement
this mechanism with a specified finite presentation.
Related work of Wang and Zhi uses the same normalization mechanism
to construct a polynomial separating matrix traces from a trace
on a group von Neumann algebra~\cite{WangZhi2026}.
We now introduce the game models, state the main results, and
explain the construction and its formalization.

\subsection{Notation and game models}
\label{subsec:statement-language}

A finite nonlocal game with two players is a tuple
\(G=(X,Y,\mathcal A,\mathcal B,\pi,V)\), where \(X,Y\) are finite
nonempty question sets, \(\mathcal A,\mathcal B\) are finite nonempty
answer sets, \(\pi\) is a probability distribution on \(X\times Y\),
and \(V:X\times Y\times\mathcal A\times\mathcal B\to[0,1]\)
is the payoff function.
The referee samples \((x,y)\) according to \(\pi\), sends one question
to each player, and receives answers \(a\in\mathcal A\) and
\(\beta\in\mathcal B\). The players may share resources in advance
but cannot communicate after receiving their questions.
A correlation is an array \(p(a,\beta\mid x,y)\) of nonnegative
conditional probabilities summing to one for each \((x,y)\).
Its expected payoff is
\[
 \omega(G,p)
 =\sum_{x\in X,\,y\in Y}\pi(x,y)
   \sum_{a\in\mathcal A,\,\beta\in\mathcal B}
      V(x,y,a,\beta)\,p(a,\beta\mid x,y).
\]
When \(V\) takes values in \(\{0,1\}\), as in our construction,
this is the success probability.

All Hilbert spaces are complex, with inner products linear in the
second variable; \(T^*\) denotes the adjoint of an operator \(T\).
A positive operator valued measurement (POVM) is a finite family
of bounded positive operators summing to the identity \(I\).
It is projective if its members are orthogonal projections.
For the fixed question and answer sets, let \(C_q\) consist of
correlations
\[
 p(a,\beta\mid x,y)
 =\langle\Omega,(P_a^x\otimes Q_\beta^y)\Omega\rangle,
\]
where \(\Omega\) is a unit vector in a finite dimensional tensor
product \(\HH_A\otimes\HH_B\), and \((P_a^x)_a\), \((Q_\beta^y)_\beta\)
are local POVMs for each question. The dimensions are unrestricted.
The set \(C_{qa}=\overline{C_q}\) is the closure in the Euclidean
space of correlation arrays. The commuting correlation set
\(C_{qc}\) consists of arrays
\[
 p(a,\beta\mid x,y)=\langle\Omega,P_a^xQ_\beta^y\Omega\rangle,
 \qquad P_a^xQ_\beta^y=Q_\beta^yP_a^x,
\]
where both families are POVMs on one Hilbert space of arbitrary
dimension and \(\Omega\) is a unit vector in that space.
In both models, the state is fixed, and each measurement depends
only on the corresponding player's question.
Finally, \(C_c\) is the convex hull of deterministic correlations
given by response functions \(X\to\mathcal A\) and
\(Y\to\mathcal B\). These are the classical correlations with
shared randomness.

For \(t\in\{c,q,qa,qc\}\), define
\[
 \omega_t(G)=\sup_{p\in C_t}\omega(G,p).
\]
The inclusion of the correlation sets and the continuity of the payoff gives
\[
 \omega_c(G)\le\omega_q(G)=\omega_{qa}(G)\le\omega_{qc}(G)\le1.
\]
We use \(\omega_q\) for the common quantum value throughout.
A strategy is perfect if its success probability is one;
a supremum of one need not be attained.
Thus \(\omega_q(G)<1=\omega_{qc}(G)\) separates
\(C_{qa}\) from \(C_{qc}\), including all limits of finite
dimensional quantum correlations.

To specify our game, write \(\N=\{1,2,\ldots\}\), let
\(\Z,\R,\C\) denote the integers, reals, and complex numbers,
and put \(\F_p=\Z/p\Z\) for a prime \(p\).
For a ring \(R\), let \(\Mat_{m,n}(R)\) denote the matrices
with \(m\) rows and \(n\) columns, and put
\(\Mat_d(R)=\Mat_{d,d}(R)\).
Consider a binary system \(A\xi=b\), with
\(A\in\Mat_{m,n}(\F_2)\), \(b\in\F_2^m\), and an unknown
\(\xi\in\F_2^n\), where \(m,n\in\N\).
Assume that every row has three nonzero entries, with support
\[
 V_r=\{e:A_{re}=1\}=\{e_{r,1}<e_{r,2}<e_{r,3}\}
 \qquad(1\le r\le m).
\]
An incidence is a pair \((r,e)\) with \(e\in V_r\).
The game \(G(A,b)\) has question sets \(X=\{1,\ldots,m\}\),
\(Y=\{1,\ldots,n\}\) and answer sets \(\mathcal A=\F_2^3\),
\(\mathcal B=\F_2\). It samples these \(3m\) incidences uniformly.
Alice receives \(r\) and answers a triple
\(a=(a_1,a_2,a_3)\in\F_2^3\); Bob receives \(e\) and answers
\(\beta\in\F_2\). For \(e=e_{r,j}\), they win precisely when
\[
 a_1+a_2+a_3=b_r\quad\text{in }\F_2,
 \qquad a_j=\beta.
\]
These are the parity and consistency conditions.
Pairs with \(e\notin V_r\) have probability zero and payoff zero.
We write \(e_j\) for the vector in \(\F_2^m\) with entry one
at position \(j\) and zero elsewhere.

The solution group connects these constraints to unitary
representations. Write \(F(S)\) for the free group on a finite
generating set \(S\), and \(\langle S\mid R\rangle\) for its
quotient imposing \(r=1\) for every relator \(r\in R\).
The solution group \(\Gamma(A,b)\) is generated by
\(x_1,\ldots,x_n,J\), with \(x_e^2=J^2=1\), with \(J\)
central, and with the generators in each row commuting and satisfying
\[
 x_{e_{r,1}}x_{e_{r,2}}x_{e_{r,3}}=J^{b_r}
 \qquad(1\le r\le m).
\]
Here \(b_r\) is read as zero or one in the exponent.
We call \(J\) the distinguished central involution, allowing
\(J=1\) until nontriviality is established.

The involution \(J\) encodes the parity sign: when it is represented
by \(-I\), the row relations become the operator form of the
binary parity constraints.

For \(d\in\N\), let \(\U(d)\) be the unitary group of
\(\Mat_d(\C)\), and let \(I_d\) be the identity matrix.
Relation errors are measured in the normalized Hilbert--Schmidt norm
\[
 \norm{X}_{2,d}=\left(\frac1d\Tr(X^*X)\right)^{1/2},
\]
where \(\Tr\) is the ordinary matrix trace.
A homomorphism \(f:F(S)\to\U(d)\) assigns a unitary matrix to
each generator; its error on a relator \(r\) is
\(\norm{f(r)-I_d}_{2,d}\).

\subsection{Main results}

Our first result gives a finite nonlocal game with a perfect
commuting operator strategy whose quantum value is strictly
less than one. This gap prevents finite dimensional quantum
correlations from approximating any perfect commuting correlation
for the game, even when the local dimensions are allowed to grow.

\begin{theorem}\label{thm:existence}
There exists a finite nonlocal game \(G\) with two players and
payoffs in \(\{0,1\}\) such that
\[
 \omega_q(G)=\omega_{qa}(G)<1=\omega_{qc}(G).
\]
In particular, \(C_{qa}\ne C_{qc}\) for its question and answer sets.
\end{theorem}


The next theorem constructs the game in Theorem~\ref{thm:existence}.
Its solution group has a nontrivial central involution whose image
under unitary assignments approaches the identity as the defining
relation errors tend to zero, uniformly in the matrix dimension.
All errors and distances are measured in the normalized
Hilbert-Schmidt norm.

\begin{theorem}
\label{thm:main}\label{thm:construction}
The construction in Sections~\ref{sec:cover}, \ref{sec:obstruction},
and~\ref{sec:translation} specifies
\[
 A\in\Mat_{m,n}(\F_2),\qquad b=e_{1{,}417{,}141},\qquad
 m=1{,}417{,}152,\quad n=1{,}889{,}684.
\]
Every row of \(A\) has exactly three nonzero entries, and
\(G=G(A,b)\) satisfies
\[
 \omega_q(G)=\omega_{qa}(G)<1=\omega_{qc}(G).
\]
Let \(\Sigma=\Gamma(A,b)=\langle S_\Sigma\mid R_\Sigma\rangle\),
where \(S_\Sigma=\{x_1,\ldots,x_n,J_\Sigma\}\) and
\(R_\Sigma\) is the standard relator list obtained from the
solution group equations by writing each \(u=v\) as \(uv^{-1}\).
Then \(J_\Sigma\ne1\). Nevertheless, for every \(\eta>0\)
, there exists \(\delta>0\) such that, for every \(d\in\N\)
and every homomorphism \(f:F(S_\Sigma)\to\U(d)\),
\begin{equation}\label{eq:main-quantifiers}
 \max_{r\in R_\Sigma}\norm{f(r)-I_d}_{2,d}\le\delta
 \quad\Longrightarrow\quad
 \norm{f(J_\Sigma)-I_d}_{2,d}<\eta.
\end{equation}
\end{theorem}


The nontriviality of \(J_\Sigma\) yields a perfect commuting strategy
by Proposition~\ref{prop:commuting}, so \(\omega_{qc}(G)=1\).
By the criterion of Slofstra and Vidick
\cite[Corollary~5.3]{SV}, stated in Proposition~\ref{thm:SV},
the equality \(\omega_q(G)=1\) would yield approximate representations
with arbitrarily small relator errors while the image of
\(J_\Sigma\) remains a fixed positive distance from the identity.
This contradicts condition~\eqref{eq:main-quantifiers}, which we call
\emph{uniform asymptotic triviality}.
Thus \(\omega_q(G)<1\), and hence \(C_{qa}\ne C_{qc}\).


The Lean development also proves the exact classical value (Proposition \ref{prop:classical}) and
explicit bounds on the quantum gap
(Remark~\ref{rem:numerical-values}):
\[
 \omega_c(G)=1-\frac1{4{,}251{,}456},\qquad
 2^{-2^{50003}}\le1-\omega_q(G)\le\frac1{4{,}251{,}456}.
\]
The lower bound is independent of the local dimensions.
The upper bound follows from \(\omega_q(G)\ge\omega_c(G)\).
\subsection{Proof strategy}

The construction first produces a nontrivial group element that
is uniformly asymptotically trivial, then transfers this
obstruction to a solution group.
Section~\ref{sec:cover} specifies a finitely presented group \(Q\)
and a subgroup \(H\le Q\), both with Kazhdan's property~\((T)\),
the rigidity condition recalled in Section~\ref{sec:preliminaries}.
The elements conjugating \(H\) into itself generate \(Q\), and a
specified element \(t\) satisfies \(tHt^{-1}<H\).
Finitely presented Kazhdan covers provide property~\((T)\), while
a Laurent polynomial matrix model proves that the inclusion is
proper.

The amalgamated double \(D=Q*_H Q\), formed by identifying \(H\)
in two copies of \(Q\), contains a specified word \(w\) of
infinite order. The normalization theorem forces every homomorphism
from \(D\) into the unitary group of a tracial matrix ultraproduct
to send \(w\) to the identity
(Proposition~\ref{prop:double-obstruction}).
Such ultraproducts encode bounded matrix sequences modulo errors tending
to zero along an ultrafilter in normalized Hilbert--Schmidt norm;
they are defined in Section~\ref{sec:preliminaries}.
Adjoining a central involution \(J\) and a generator \(z\) with
\(zwz^{-1}=wJ\) produces a finitely presented group \(\Lambda\).
A normal form argument shows that \(J\ne1\), while the relation
forces every such ultraproduct homomorphism to send \(J\) to the
identity (Proposition~\ref{prop:central-obstruction}).

Slofstra's involution substitution and wagon wheel constructions
\cite{Slofstra} give the specified binary system and an embedding
\(\iota:\Lambda\to\Sigma=\Gamma(A,b)\) with
\(\iota(J)=J_\Sigma\), so \(J_\Sigma\ne1\)
(Proposition~\ref{prop:solution-embedding}).
Every homomorphism from \(\Sigma\) into the unitary group of a
tracial matrix ultraproduct sends \(J_\Sigma\) to the identity,
because its composition with \(\iota\) sends \(J\) to the identity.
Lemma~\ref{lem:ultraproduct} converts this vanishing property into
\eqref{eq:main-quantifiers}.
This completes the passage from the group obstruction to the
separation in Theorems~\ref{thm:main} and~\ref{thm:existence}.

\subsection{Formalization and its relation to the proof}
\label{subsec:formalization}


The accompanying development in Lean~4~\cite{deMouraUllrich2021}%
\footnote{Source repository:
\url{https://github.com/leanhnhe091221-boop/explicit-linear-system-game-formalization}.}
It specifies the finite presentations, word substitutions, binary
system, and game. Using mathlib~\cite{mathlib2020}, it proves
the main separation, the group assertions in
Theorem~\ref{thm:main}, the exact classical value, and the explicit
lower bound on the quantum gap.
The formal statements use the same strategy models as the paper:
arbitrary POVMs, all finite local dimensions for quantum strategies,
and Hilbert spaces of arbitrary dimension for commuting strategies.
Likewise, \(C_{qa}\) is defined as the closure of \(C_q\).

The final results use no admitted proofs or unproved mathematical
hypotheses. The formalization follows the paper's construction,
with some intermediate results formalized only in the form needed
for the final conclusions.
The repository fixes the Lean and mathlib versions and provides
build instructions and axiom audits.
Its \texttt{README} and Appendix~\ref{app:formalization} document
the correspondence with the manuscript, the precise scope of
the intermediate formalizations, and the verification commands.

\subsection{Organization of the paper}

Section~\ref{sec:preliminaries} introduces approximate representations,
tracial matrix ultraproducts, and normalization.
Sections~\ref{sec:cover} and~\ref{sec:obstruction} construct the
Kazhdan group pair and the central involution.
Section~\ref{sec:translation} specifies the binary system, and
Section~\ref{sec:games} proves the separation of game values.
Appendix~\ref{app:data} records the presentation counts and
reproducibility data, and Appendix~\ref{app:formalization} describes
the correspondence with Lean.
\section{Preliminaries}\label{sec:preliminaries}

This section provides the approximation and rigidity results used
in the construction. We relate uniform asymptotic triviality to
homomorphisms into tracial matrix ultraproducts, then explain how
property~\((T)\) and the results in \cite{ALT,Thom} give the
normalization theorem needed in
Section~\ref{sec:obstruction}.

\subsection{Conventions}\label{subsec:conventions}

We retain the notation of Section~\ref{subsec:statement-language}.
All groups are discrete, and their identity is denoted by \(1\).
We write \(\langle S\rangle\) for the subgroup generated by \(S\),
\(H\le Q\) for subgroup inclusion, \(H<Q\) for proper inclusion,
and \(N\triangleleft Q\) for a normal subgroup. Our commutator
convention is
\[
 [a,b]=aba^{-1}b^{-1}.
\]
The normal closure \(\langle\!\langle R\rangle\!\rangle\) of a
relator list \(R\) consists of finite products
of conjugates of its elements and their inverses.
A map on the generators of \(\langle S\mid R\rangle\) extends to
a group homomorphism exactly when every relator maps to the
identity; see \cite{LS}. We continue to allow the distinguished
central involution to be the identity until nontriviality is proved.

For a Hilbert space \(\mathcal K\), let \(\U(\mathcal K)\) denote
its group of unitary operators. A unitary representation of a group
\(Q\) is a homomorphism \(Q\to\U(\mathcal K)\).
We write \(\norm{X}\) for the operator norm and
\(\tr_d=d^{-1}\Tr\) for the normalized matrix trace. Thus
\[
 \norm{X}_{2,d}=\tr_d(X^*X)^{1/2}
 \qquad(X\in\Mat_d(\C)).
\]
Matrix and generator indices start at one, except for word positions
explicitly numbered from zero.

\subsection{Matrix ultraproducts and approximate representations}
\label{subsec:ultraproducts}

Tracial matrix ultraproducts turn approximate matrix relations into
exact relations. We use the conventions of \cite[Section~2.1]{ALT}.
Fix a nonprincipal ultrafilter \(\omega\) on \(\N\): it is a
collection of subsets closed under finite intersections and
supersets, containing no finite set and containing exactly one of
each set and its complement. A set in \(\omega\) is called
\(\omega\)-large. For a bounded complex sequence, the ultralimit
\(\lim_{k\to\omega}z_k=z\) is characterized by
\(\{k:|z_k-z|<\varepsilon\}\in\omega\) for every
\(\varepsilon>0\). It exists uniquely and agrees with the ordinary
limit whenever that limit exists.

For positive integers \(d_k\), with no growth assumption, put
\[
 \begin{aligned}
 \ell^\infty(\Mat_{d_k})
 &=\{(a_k):a_k\in\Mat_{d_k}(\C),\ \sup_k\norm{a_k}<\infty\},\\
 \mathcal J_\omega
 &=\{(a_k)\in\ell^\infty(\Mat_{d_k}):
                 \lim_{k\to\omega}\norm{a_k}_{2,d_k}=0\}.
 \end{aligned}
\]
The set \(\mathcal J_\omega\) is a closed two sided \(*\)-ideal.
The tracial matrix ultraproduct is
\begin{equation}\label{eq:ultraproduct}
 \M=\prod_\omega(\Mat_{d_k}(\C),\tr_{d_k})
   =\ell^\infty(\Mat_{d_k})/\mathcal J_\omega.
\end{equation}
Writing \([(a_k)]\) for a class in this quotient, its trace and
associated norm satisfy
\[
 \tau([(a_k)])=\lim_{k\to\omega}\tr_{d_k}(a_k),\qquad
 \norm{[(a_k)]}_2
   =\tau([(a_k)]^*[(a_k)])^{1/2}
   =\lim_{k\to\omega}\norm{a_k}_{2,d_k}.
\]
In particular, \(\norm{x}_2=0\) exactly when \(x=0\).
In its representation by left multiplication on the trace Hilbert
space \(L^2(\M,\tau)\), the ultraproduct is a finite von Neumann
algebra with faithful normal trace \(\tau\).
Here \(L^2(\M,\tau)\) is the completion for the inner product
\(\langle x,y\rangle=\tau(x^*y)\).
We write \(\U(\M)=\{u\in\M:u^*u=uu^*=1\}\).
For \(S\subseteq\M\), its relative commutant is
\[
 S'\cap\M=\{x\in\M:xs=sx\text{ for every }s\in S\}.
\]

We measure approximate representations by their defining relator
errors, following \cite[Definition~2.1]{SV}. The uniformity in the
next definition is needed to exclude approximation in growing
matrix dimensions.

\begin{definition}\label{def:approx}
Let \(P=\langle S\mid R\rangle\) be a finite presentation and
\(\delta\ge0\). A \(\delta\)-representation is a homomorphism
\(f:F(S)\to\U(d)\), for some \(d\in\N\), whose maximum relator error satisfies
\[
 \operatorname{err}_R(f)
 :=\max_{r\in R}\norm{f(r)-I_d}_{2,d}\le\delta.
\]
with \(\defect_R(f)=0\) when \(R\) is empty.
A word \(v\in F(S)\) is uniformly asymptotically trivial if for
every \(\eta>0\) there is \(\delta>0\) such that every
\(\delta\)-representation, in every dimension, satisfies
\(\norm{f(v)-I_d}_{2,d}<\eta\).
\end{definition}

The following standard ultraproduct criterion connects exact
homomorphisms with uniform asymptotic triviality. We include a
short proof to make the uniformity in the matrix dimension explicit.

\begin{lemma}\label{lem:ultraproduct}
Let \(P=\langle S\mid R\rangle\) be a finite presentation and
\(v\in F(S)\). The following are equivalent:
\begin{enumerate}
\item every homomorphism \(P\to\U(\M)\), for every tracial matrix
ultraproduct \(\M\), sends the image of \(v\) to \(1\);
\item \(v\) is uniformly asymptotically trivial.
\end{enumerate}
\end{lemma}

\begin{proof}
Suppose (ii) fails. There are \(\eta_0>0\), dimensions \(d_k\),
and homomorphisms \(f_k:F(S)\to\U(d_k)\) such that
\[
 \defect_R(f_k)\le1/k,\qquad
 \norm{f_k(v)-I_{d_k}}_{2,d_k}\ge\eta_0.
\]
The unitary classes \([(f_k(s))]\), for \(s\in S\), satisfy every
relator and therefore define \(\pi:P\to\U(\M)\), with \(\M\)
as in \eqref{eq:ultraproduct}. Since
\[
 \norm{\pi(v)-1}_2
   =\lim_{k\to\omega}\norm{f_k(v)-I_{d_k}}_{2,d_k}\ge\eta_0,
\]
condition (i) fails. This proves (i)\(\Rightarrow\)(ii).

Conversely, assume (ii) and let \(\pi:P\to\U(\M)\).
Every unitary \(u=[(a_k)]\) in \(\M\) has unitary matrix
representatives. Indeed, completing the polar part of each \(a_k\)
to a unitary \(u_k\) and using \(|t-1|\le|t^2-1|\) for
\(t\ge0\), we obtain
\[
 \norm{a_k-u_k}_{2,d_k}
 \le\norm{a_k^*a_k-I_{d_k}}_{2,d_k}
 \underset{k\to\omega}{\longrightarrow}0.
\]
Thus \([(u_k)]=u\). Choose such representatives for each
\(\pi(s)\) and extend them to homomorphisms
\(f_k:F(S)\to\U(d_k)\). Each relator error tends to zero along
\(\omega\), and finiteness of \(R\) implies the same for
\(\defect_R(f_k)\). By (ii),
\(\norm{f_k(v)-I_{d_k}}_{2,d_k}\) also tends to zero along
\(\omega\). Taking ultralimits gives \(\norm{\pi(v)-1}_2=0\),
and hence \(\pi(v)=1\).
\end{proof}

The composition argument below will transfer the obstruction to
the solution group in Section~\ref{sec:games}.

\begin{remark}\label{rem:embedding-direction}
Let \(\iota:P\to Q\) be a homomorphism and \(v\in P\).
If every homomorphism from \(P\) into the unitary group of a
tracial matrix ultraproduct sends \(v\) to \(1\), then every such
homomorphism \(\rho\) from \(Q\) satisfies
\(\rho(\iota(v))=(\rho\circ\iota)(v)=1\).
If \(\iota\) is injective and \(v\ne1\), then \(\iota(v)\ne1\).
For finitely presented groups, Lemma~\ref{lem:ultraproduct}
converts the vanishing conclusion into uniform asymptotic
triviality of a word representing \(\iota(v)\).
No extension of approximate representations is required.
\end{remark}

\subsection{Kazhdan groups and normalization}
\label{subsec:normalization}

We recall the property~\((T)\) facts needed for the group pair
constructed in Section~\ref{sec:cover}; see \cite{BHV}.
A group \(Q\) has property~\((T)\), or is Kazhdan, if every
unitary representation \(\rho:Q\to\U(\mathcal K)\) with almost
invariant unit vectors has
a nonzero invariant vector. Here almost invariance means that
for every nonempty finite \(F\subset Q\) and \(\varepsilon>0\), some unit
vector \(\xi\) satisfies
\(\max_{q\in F}\norm{\rho(q)\xi-\xi}<\varepsilon\).
Equivalently, there are a nonempty finite \(F\subset Q\) and
\(\kappa>0\) such that every unitary representation
\(\rho:Q\to\U(\mathcal K)\) satisfies
\[
 \max_{q\in F}\norm{\rho(q)\xi-\xi}\ge\kappa\norm{\xi}
 \qquad(\xi\perp\mathcal K^Q),
\]
where \(\mathcal K^Q=\{\xi:\rho(q)\xi=\xi\text{ for all }q\in Q\}\)
is the invariant subspace.
Discrete Kazhdan groups are finitely generated, and property~\((T)\)
passes to quotients and is stable under extensions: if
\(N\triangleleft Q\) and both \(N\) and \(Q/N\) have property~\((T)\),
then so does \(Q\).
We also use property~\((T)\) of \(\SL_3(\Z)\), the group of integral
\(3\times3\) matrices of determinant one~\cite{BHV}.

The subgroup condition used with property~\((T)\) records which
conjugations carry a subgroup into itself. We follow
\cite[Introduction]{KunThom2026}.

\begin{definition}\label{def:infranormal}
For \(H\le Q\), its compression semigroup is
\[
 P_H=\{t\in Q:tHt^{-1}\le H\}.
\]
The subgroup \(H\) is infranormal in \(Q\) if
\(\langle P_H\rangle=Q\), where generation is as a group and
therefore allows inverses. An element \(t\in P_H\) is a strict
compressor if \(tHt^{-1}<H\).
\end{definition}

We now explain the analytic input that turns these group properties
into normalization. A subalgebra of \(\M\) is internal if it has
the form
\[
 [B_k]_\omega
 =\{[(b_k)]:b_k\in B_k,\ \sup_k\norm{b_k}<\infty\},
\]
where each \(B_k\subseteq\Mat_{d_k}(\C)\) is a unital
\(*\)-subalgebra.
For a unital von Neumann subalgebra \(C\subseteq\M\), let
\(E_C:\M\to C\) be the conditional expectation preserving the
trace. Its extension to \(L^2(\M,\tau)\) is the orthogonal
projection onto \(L^2(C)\), the closure of \(C\) in that Hilbert
space.

Fix \(\pi:Q\to\U(\M)\) and put \(C=\pi(H)'\cap\M\).
The conjugation action
\(\rho_H(h)x=\pi(h)x\pi(h)^*\) extends to a unitary
representation on \(L^2(\M,\tau)\), with invariant subspace
\(L^2(C)\); see \cite[Section~2.1]{ALT}.
If \(H\) has property~\((T)\), choose a nonempty finite generating
set \(h_1,\ldots,h_\ell\) containing a Kazhdan set and put
\(u_j=\pi(h_j)\).
Applying the Kazhdan inequality on the orthogonal complement of
\(L^2(C)\) gives a constant \(\gamma\in(0,1)\) such that
\[
 \frac1{4\ell}\sum_{j=1}^{\ell}\norm{u_jx-xu_j}_2^2
   \ge\gamma\norm{x-E_Cx}_2^2\qquad(x\in\M).
\]
This is spectral gap off the commutant, in the normalization of
\cite[(1.1)]{ALT}.

By \cite[Theorem~B]{ALT}, the relative commutant of any finite
unitary tuple with this spectral gap is internal. Thus the
commutant of every image of a Kazhdan group is internal.
This supplies the centralizer hypothesis of
\cite[Theorem~1.2]{Thom}, yielding the following result.

\begin{theorem}\label{thm:normalization}
Let \(H\le Q\) be Kazhdan groups, with \(H\) infranormal in \(Q\).
For every homomorphism \(\pi:Q\to\U(\M)\) into a tracial matrix
ultraproduct,
\[
 \pi(q)\bigl(\pi(H)'\cap\M\bigr)\pi(q)^{-1}
   =\pi(H)'\cap\M\qquad(q\in Q).
\]
\end{theorem}

Thus each \(\pi(q)\) normalizes the relative commutant: conjugation
carries that algebra onto itself. The homomorphism need not be
injective. Property~\((T)\) is required for both \(H\) and \(Q\);
these hypotheses are verified in Section~\ref{sec:cover}.

\section{A Kazhdan pair with strict compression}\label{sec:cover}

We construct a finitely presented group $Q$ on $48$ generators and a
subgroup $H$ generated by $24$ specified elements. Both groups have
property~(T), and $H$ is infranormal in $Q$. We also specify
$t\in Q$ and $h\in H$ such that
\[
 tHt^{-1}<H,\qquad h\notin tHt^{-1}.
\]
These are the hypotheses needed for the obstruction in
Section~\ref{sec:obstruction}. The construction uses the finitely
presented Kazhdan covers of Ershov and Jaikin-Zapirain, while a
Laurent matrix representation proves strict compression. We retain
all generator and relator choices because they determine the binary
linear system in Section~\ref{sec:translation}.

Throughout this section, put
\[
 I=\{1,2,3\},\qquad E=\{(i,j)\in I^2:i\ne j\}.
\]
We call the six pairs in $E$ root positions. For a unital ring $R$,
let $E_{ij}$ be the matrix with entry $1$ at $(i,j)$ and zeros
elsewhere, and set
\[
 e_{ij}(f)=I_3+fE_{ij}\qquad(f\in R,\ (i,j)\in E).
\]
Here $I_3$ is the identity matrix. These elementary matrices generate
$\EL_3(R)$ inside $\GL_3(R)$, the group of invertible $3\times3$
matrices over $R$. In particular, $e_{ij}(f)^{-1}=e_{ij}(-f)$.
We retain the commutator and presentation conventions of
Section~\ref{subsec:conventions}.

\subsection{Finite presentations over \texorpdfstring{$\F_5$}{F5}}

Let $\F_5=\Z/5\Z$. For an integer $d\ge1$, define $\Delta_d$ by the generators
\[
 Y_{ij}(m),\qquad (i,j)\in E,\quad 0\le m\le d,
\]
and the five relation families
\begin{align*}
 Y_{ij}(m)^5&=1 &&\text{(E0)},\\
 [Y_{ij}(m),Y_{i'j'}(m')]&=1
 &&\text{if }\{i,i'\}\cap\{j,j'\}=\varnothing
 \quad\text{(E1)},\\
 [Y_{ij}(m),Y_{jk}(0)]&=Y_{ik}(m)
 &&\text{if }i,j,k\text{ are distinct}\quad\text{(E2)},\\
 [Y_{ij}(0),Y_{jk}(m)]&=Y_{ik}(m)
 &&\text{if }i,j,k\text{ are distinct}\quad\text{(E3)},\\
 [[Y_{ij}(m),Y_{jk}(m')],Y_{i'k'}(m'')]&=1
 &&\text{(E4)}.
\end{align*}
In (E4), the conditions are
\[
 \{i,i'\}\cap\{k,k'\}=\varnothing,
 \qquad j\notin\{i,k\}.
\]
All root indices range over $I$, and each displayed root symbol has
unequal indices. Coefficient labels range independently over
$\{0,\ldots,d\}$ unless fixed to zero. Every allowed instance is
retained, including repetitions and redundant relations. An equation
$u=v$ contributes the relator $uv^{-1}$. These conventions give a
finite presentation on $6(d+1)$ generators. The label $m$ is an
index; in the matrix maps below, label $0$ represents the coefficient $1$.

The following specialization of \cite[Theorem~6.9]{EJZ} supplies
property~(T) for these presented groups. A Kazhdan cover here means
a Kazhdan group equipped with a surjective homomorphism to the target.
The coefficient algebras are associative and need not be commutative.

\begin{proposition}\label{prop:finite-cover}
For every $d\ge1$, the group $\Delta_d$ is finitely presented and
Kazhdan. If an associative unital $\F_5$-algebra $R$ is generated as a unital
algebra by $a_1,\ldots,a_d$, there is a surjective homomorphism
\[
 \Delta_d\longrightarrow\EL_3(R),\qquad
 Y_{ij}(0)\longmapsto e_{ij}(1),\qquad
 Y_{ij}(m)\longmapsto e_{ij}(a_m)\quad(1\le m\le d).
\]
\end{proposition}

\begin{proof}
Apply \cite[Theorem~6.9]{EJZ} with $n=3$, base field $\F_5$, and
the free associative algebra $\F_5\langle z_1,\ldots,z_d\rangle$.
The field has cardinality $5$, as required when $n<7$. Its basis
over the prime field is $\{1\}$, so the defining relations
in that theorem specialize to (E0)--(E4).
Family (E5) is tautological: for distinct $i,k\in I$ there is only
one possible intermediate index. Family (E6) applies only to the
base ring $\Z$. Thus the cover in that theorem is exactly
$\Delta_d$, proving property~(T).

For general $R$, the substitution $z_m\mapsto a_m$ is a surjective
algebra homomorphism. Applied to matrix entries, it induces a
surjection between elementary groups, since every coefficient in
$R$ has a preimage. Composing with the cover of the free algebra
gives the asserted map.
\end{proof}

We use $\Delta_6$ for a normal subgroup of $Q$ and $\Delta_3$ for
$H$. Realizing these subgroups as quotient images of the covers
will establish their property~(T).

\subsection{The presentation of \texorpdfstring{$Q$}{Q}}

Fix the ordered coefficient lists
\[
 M=(0,1,2,3,-1,-2,-3),\qquad M^+=(0,1,2,3).
\]
Membership refers to the underlying sets; the order specifies the
numbering of generators. The group $Q$ has $48=6\cdot7+6$ generators:
\[
 X_{ij}(m)\quad((i,j)\in E,\ m\in M),\qquad
 s_{ab}\quad((a,b)\in E).
\]
Negative labels name distinct generators: $X_{ij}(-c)$ is not the
notation for $X_{ij}(c)^{-1}$. The $s_{ab}$ are called shear
generators because their matrix images will act by shears on
monomial exponents.

First impose the relations of $\Delta_6$, renaming its six nonzero
labels as $1,2,3,-1,-2,-3$:
\begin{align}
 X_{ij}(m)^5&=1,
 \label{eq:q-power}\\
 [X_{ij}(m),X_{i'j'}(m')]&=1
 \quad\text{if }\{i,i'\}\cap\{j,j'\}=\varnothing,\\
 [X_{ij}(m),X_{jk}(0)]&=X_{ik}(m)
 \quad\text{if }i,j,k\text{ are distinct},\\
 [X_{ij}(0),X_{jk}(m)]&=X_{ik}(m)
 \quad\text{if }i,j,k\text{ are distinct},\\
 [[X_{ij}(m),X_{jk}(m')],X_{i'k'}(m'')]&=1.
 \label{eq:q-triple}
\end{align}
In \eqref{eq:q-triple}, require
$\{i,i'\}\cap\{k,k'\}=\varnothing$ and $j\notin\{i,k\}$.
The root indices range over $I$, each root symbol has unequal
indices, and $m,m',m''$ range independently over $M$.

The shear generators satisfy
\begin{align}
 [s_{ab},s_{cd}]&=1
 \quad\text{if }b\ne c\text{ and }a\ne d,
 \label{eq:shear-comm}\\
 [s_{ab},s_{bc}]&=s_{ac}
 \quad\text{if }a,b,c\text{ are distinct},\\
 (s_{12}s_{21}^{-1}s_{12})^4&=1.
 \label{eq:shear-torsion}
\end{align}
Here all pairs indexing generators belong to $E$.
These relations present $\SL_3(\Z)$ under
$s_{ab}\mapsto I_3+E_{ab}$; see \cite[Section~5.6]{Soule}.
We will use this presentation to identify a quotient of $Q$ as a
quotient of $\SL_3(\Z)$.

For the action relations, fix $(a,b),(i,j)\in E$ and let $k$ be
the unique element of $I\setminus\{i,j\}$. For both
$\varepsilon\in\{1,-1\}$, impose
\begin{align}
 s_{ab}^{\varepsilon}X_{ij}(0)s_{ab}^{-\varepsilon}
 &=X_{ij}(0),
 \label{eq:action-zero}\\
 s_{ab}^{\varepsilon}X_{ij}(\sigma c)s_{ab}^{-\varepsilon}
 &=X_{ij}(\sigma c)
 \quad(c\ne b,\ \sigma\in\{1,-1\}),\\
 s_{ab}^{\varepsilon}X_{ij}(\sigma b)s_{ab}^{-\varepsilon}
 &=[X_{ik}(\sigma b),X_{kj}(\varepsilon\sigma a)]
 \quad(\sigma\in\{1,-1\}).
 \label{eq:action-mixed}
\end{align}
Here $c\in I$, and products such as $\varepsilon\sigma a$ are
coefficient labels in $M$.

Finally, let $\nu:I\to I$ be the cyclic successor map
$1\mapsto2\mapsto3\mapsto1$, and define the words
\[
 r_c=(s_{c,\nu(c)}s_{\nu(c),c}^{-1}s_{c,\nu(c)})^2
 \qquad(c\in I).
\]
Impose
\begin{equation}\label{eq:negative-recovery}
 r_cX_{ij}(c)r_c^{-1}=X_{ij}(-c)
 \qquad((i,j)\in E,\ c\in I).
\end{equation}
The words $r_c$ introduce no new generators. Relations
\eqref{eq:q-power}--\eqref{eq:negative-recovery}, with every
allowed index instance retained, define $Q$.

Define
\begin{equation}\label{eq:NH-def}
 \begin{aligned}
 N&=\langle X_{ij}(m):(i,j)\in E,\ m\in M\rangle,\\
 H&=\langle X_{ij}(m):(i,j)\in E,\ m\in M^+\rangle.
 \end{aligned}
\end{equation}
Thus $N$ has $42$ specified generators and $H$ has $24$.
For subsequent word constructions, order root pairs as
$(1,2),(1,3),(2,1),(2,3),(3,1),(3,2)$.
List all $X$ generators in this order of pairs and, within each
pair, in the displayed order of $M$. Then append the six shear
generators in the same order of pairs.

\subsection{Property~(T) and infranormality}

The two covers give property~(T) for $N$ and $H$. The action
relations make $N$ normal and the shears compress $H$; the recovery
relations then establish infranormality.

\begin{proposition}\label{prop:kazhdan-pair}
The subgroup $N$ is normal in $Q$, and $N$, $H$, and $Q$ have
property~(T). Moreover, $H$ is infranormal in $Q$, with
\[
 s_{ab}Hs_{ab}^{-1}\le H\quad((a,b)\in E),\qquad
 Q=\langle H,s_{ab}:(a,b)\in E\rangle.
\]
\end{proposition}

\begin{proof}
Relations \eqref{eq:q-power}--\eqref{eq:q-triple} give surjections
$\Delta_6\to N$ and $\Delta_3\to H$. The first uses the label
renaming above, and the second restricts to labels in $M^+$.
Thus $N$ and $H$ have property~(T) by
Proposition~\ref{prop:finite-cover} and permanence under quotients.

For both conjugating signs, relations
\eqref{eq:action-zero}--\eqref{eq:action-mixed} express conjugates
of the generators of $N$ as words in $N$. Hence each shear
normalizes $N$. Since the remaining generators lie in $N$, we have
$N\triangleleft Q$. The quotient $Q/N$ is generated by the shear
images and satisfies \eqref{eq:shear-comm}--\eqref{eq:shear-torsion},
so it is a quotient of $\SL_3(\Z)$. Property~(T) for this group
and permanence under quotients and extensions now give property~(T)
for $Q$; see \cite[Theorem~1.3.4, Example~1.7.4(i), and
Proposition~1.7.6]{BHV}.

With $\varepsilon=1$, the action relations send each generator
of $H$ into $H$. In the only case where a positive label changes,
\[
 s_{ab}X_{ij}(b)s_{ab}^{-1}
 =[X_{ik}(b),X_{kj}(a)]\in H.
\]
Thus every shear belongs to the compression semigroup $P_H$.
Relation~\eqref{eq:negative-recovery} expresses each generator
with a negative label using $H$ and shear words, giving
$Q=\langle H,s_{ab}:(a,b)\in E\rangle$. Since $H\subseteq P_H$,
these compressors generate $Q$ as a group. This proves infranormality.
\end{proof}

\subsection{A Laurent matrix representation and strict compression}

We prove strict compression by detecting subgroup nonmembership
in a matrix representation. Set
\[
 R=\F_5[x_1^{\pm1},x_2^{\pm1},x_3^{\pm1}],\qquad
 R_+=\F_5[x_1,x_2,x_3],
\]
where $x_1,x_2,x_3$ are commuting indeterminates. The Laurent
monomials in $R$ allow integer exponents; those in $R_+$ have
nonnegative exponents. For a column vector
$v=(v_1,v_2,v_3)^{\mathsf T}\in\Z^3$, where $\mathsf T$ denotes
transpose, write $x^v=x_1^{v_1}x_2^{v_2}x_3^{v_3}$.
For $A\in\GL_3(\Z)$,
define an $\F_5$-linear map by
\[
 \alpha_A(x^v)=x^{Av}.
\]
Here $A$ is a local integer action matrix, distinct from the binary
constraint matrix used later. Since exponents add under multiplication,
$\alpha_A$ is an algebra automorphism, with
$\alpha_A\alpha_{A'}=\alpha_{AA'}$ and inverse $\alpha_{A^{-1}}$.
Applying this action to matrix entries gives the semidirect product
\[
 K=\GL_3(R)\rtimes_\alpha\GL_3(\Z),\qquad
 (g,A)(g',A')=(g\alpha_A(g'),AA').
\]
In particular,
\begin{equation}\label{eq:semidirect-conjugation}
 (I_3,A)(g,I_3)(I_3,A)^{-1}=(\alpha_A(g),I_3).
\end{equation}
The identity matrices in the two components are over $R$ and $\Z$,
respectively. Write $T_{ab}=I_3+E_{ab}\in\SL_3(\Z)$; it adds
coordinate $b$ to coordinate $a$ in an integer column vector.

The next representation realizes each signed coefficient label as
a Laurent monomial and each shear as its exponent action.

\begin{proposition}\label{prop:Laurent-quotient}
There is a homomorphism $\psi:Q\to K$ determined by
\begin{align*}
 X_{ij}(0)&\longmapsto(e_{ij}(1),I_3),\\
 X_{ij}(\sigma c)&\longmapsto(e_{ij}(x_c^\sigma),I_3)
 \quad(c\in I,\ \sigma\in\{1,-1\}),\\
 s_{ab}&\longmapsto(I_3,T_{ab}).
\end{align*}
\end{proposition}

\begin{proof}
The root generator images satisfy
\eqref{eq:q-power}--\eqref{eq:q-triple} by
Proposition~\ref{prop:finite-cover}, applied to $R$ with algebra
generators $x_1,x_2,x_3,x_1^{-1},x_2^{-1},x_3^{-1}$.
The matrices $T_{ab}$ satisfy the shear relations by the cited
presentation of $\SL_3(\Z)$.

For $\varepsilon,\sigma\in\{1,-1\}$, the exponent action gives
\begin{equation}\label{eq:Laurent-shear-action}
 \alpha_{T_{ab}^{\varepsilon}}(x_c^\sigma)=
 \begin{cases}
 x_c^\sigma,&c\ne b,\\
 x_b^\sigma x_a^{\varepsilon\sigma},&c=b.
 \end{cases}
\end{equation}
Constants are fixed. Together with
\eqref{eq:semidirect-conjugation} and the elementary matrix identity
for $f,g\in R$,
\[
 [e_{ik}(f),e_{kj}(g)]=e_{ij}(fg)
 \qquad(i,j,k\text{ distinct}),
\]
this verifies all action relations
\eqref{eq:action-zero}--\eqref{eq:action-mixed}.

On coordinates $c,\nu(c)$, the matrix
$T_{c,\nu(c)}T_{\nu(c),c}^{-1}T_{c,\nu(c)}$ is
$\left(\begin{smallmatrix}0&1\\-1&0\end{smallmatrix}\right)$.
Its square is the diagonal matrix $D_c$ with entries $-1$ in
those two coordinates and $1$ in the remaining coordinate.
Thus the shear word $r_c$ maps to $(I_3,D_c)$, and
$\alpha_{D_c}(x_c)=x_c^{-1}$ verifies
\eqref{eq:negative-recovery}. Every defining relation holds, so
the assignment induces $\psi$.
\end{proof}

Under the shear $T_{12}$, all polynomial coefficients lie in the
proper subring $\F_5[x_1,x_1x_2,x_3]$. This provides an explicit
element outside the compressed subgroup.

\begin{proposition}\label{prop:strict-compression}
With
\[
 t=s_{12}\in Q,\qquad h=X_{12}(2)\in H,
\]
one has $tHt^{-1}<H$ and $h\notin tHt^{-1}$.
\end{proposition}

\begin{proof}
Proposition~\ref{prop:kazhdan-pair} gives $tHt^{-1}\le H$.
The images of the generators of $H$ and their inverses have trivial
second component and matrix entries in $R_+$. Products preserve
this property, so every element of $\psi(H)$ has the same form.
By \eqref{eq:Laurent-shear-action}, all matrix entries in the first
component of $\psi(tHt^{-1})$ therefore lie in
\[
 S=\alpha_{T_{12}}(R_+)=\F_5[x_1,x_1x_2,x_3].
\]
Every monomial occurring in $S$ has the form
$x_1^{u+v}x_2^v x_3^w$ with nonnegative integers $u,v,w$.
Its $x_1$ exponent is at least its $x_2$ exponent. By linear
independence of Laurent monomials, $x_2\notin S$.
But $\psi(h)=(e_{12}(x_2),I_3)$ has entry $x_2$ at position
$(1,2)$ in its first component. Hence
$\psi(h)\notin\psi(tHt^{-1})$, and thus $h\notin tHt^{-1}$.
Since $h\in H$, the inclusion is proper. This argument does not
require $\psi$ to be injective.
\end{proof}

\section{From compression to a central involution}\label{sec:obstruction}

We use the group pair from Section~\ref{sec:cover} to construct a
finitely presented group \(\Lambda\) with a nontrivial central
involution \(J\) that every homomorphism into the unitary group
of a tracial matrix ultraproduct sends to the identity. There are two steps. An
amalgamated double first produces an element of infinite order
with this vanishing property. An HNN extension then converts it
into the required central involution. We prove both steps for a
general pair before specifying the presentation used in
Section~\ref{sec:translation}.

\subsection{The amalgamated double}

For \(H\le Q\), the amalgamated double \(Q*_H Q\) is generated
by two copies of \(Q\), with the corresponding elements of \(H\)
identified. Its universal property says that homomorphisms from
the two factors into a group extend uniquely to the double
exactly when they agree on \(H\).
The normal form theorem embeds both factors in the double and
ensures that a nonempty product of elements outside \(H\),
alternating between the two factors, is nonidentity
\cite[Chapter~IV, Section~2]{LS}. Such a product is called reduced.

We follow the argument in the proof of \cite[Theorem~1.3]{Thom},
using Theorem~\ref{thm:normalization}. We also verify that the
resulting word has infinite order, as required for the next extension.

\begin{proposition}\label{prop:double-obstruction}
Let \(Q\) be finitely presented, let \(H\le Q\) be finitely generated,
and suppose that both groups have property~\((T)\).
Assume \(H\) is infranormal in \(Q\) and choose
\[
 t\in Q,\quad tHt^{-1}<H,\qquad h\in H\setminus tHt^{-1}.
\]
In \(D=Q*_H Q\), write \(t_1,t_2\) for the copies of \(t\),
and identify \(h\) with its common image in the two factors. Set
\[
 w=[h,t_2t_1^{-1}].
\]
Then \(D\) is finitely presented, \(w\) has infinite order, and
every homomorphism \(D\to\U(\M)\) into the unitary group of a
tracial matrix ultraproduct sends \(w\) to the identity.
\end{proposition}

\begin{proof}
Choose a finite presentation \(Q=\langle S_Q\mid R_Q\rangle\)
and words \(v_1,\ldots,v_s\) in \(S_Q\) representing generators of \(H\).
Take two copies of this presentation and add the relations
\[
 v_j(S_Q^{(1)})=v_j(S_Q^{(2)})\qquad(1\le j\le s),
\]
where \(S_Q^{(i)}\) denotes the generator set in copy \(i\).
Agreement on these generators implies agreement on all of \(H\).
The universal property therefore identifies the resulting finitely
presented group with \(D\). Only finite generation of \(H\) is needed.

Expanding the commutator gives
\[
 w=h t_2t_1^{-1}h^{-1}t_1t_2^{-1},\qquad
 t_2^{-1}wt_2=(t_2^{-1}ht_2)(t_1^{-1}h^{-1}t_1).
\]
Since \(h\notin tHt^{-1}\), we have \(t^{-1}ht\notin H\).
Its inverse \(t^{-1}h^{-1}t\) is also outside \(H\).
Thus the two factors in the conjugate above lie outside \(H\)
and belong to different copies of \(Q\). For every integer
\(k\ge1\), the \(k\)-th power of this conjugate is a reduced product
of \(2k\) alternating factors, so it is nonidentity by the normal
form theorem. Hence \(w\)
has infinite order.

Now let \(\pi:D\to\U(\M)\) be any homomorphism and put
\(C=\pi(H)'\cap\M\). Set \(c=t_1^{-1}t_2\).
For \(a\in H\), compression gives \(tat^{-1}\in H\), so its
copies \(t_1at_1^{-1}\) and \(t_2at_2^{-1}\) coincide in \(D\).
Consequently
\[
 cac^{-1}
 =t_1^{-1}(t_2at_2^{-1})t_1
 =t_1^{-1}(t_1at_1^{-1})t_1=a.
\]
Thus \(c\) commutes with \(H\), and \(\pi(c)\in C\).
Since \(Q,H\) are Kazhdan and \(H\) is infranormal,
Theorem~\ref{thm:normalization} applied to the first factor gives
\(\pi(t_1)C\pi(t_1)^{-1}=C\). Hence
\[
 \pi(t_2t_1^{-1})
 =\pi(t_1ct_1^{-1})=\pi(t_1)\pi(c)\pi(t_1)^{-1}\in C.
\]
This operator commutes with \(\pi(h)\), giving \(\pi(w)=1\).
\end{proof}

\subsection{The central involution}

We follow the construction in the proof of
\cite[Corollary~3.2]{Slofstra} to turn \(w\) into the central
involution needed for the linear system construction.
The relation \(zwz^{-1}=wJ\) will force \(J\) to vanish whenever
\(w\) does. An HNN extension ensures that \(J\) remains
nontrivial in the group.


We use the following special case of an HNN extension.
For a group \(B\), a subgroup \(E\le B\), and an automorphism
\(\alpha:E\to E\), it is the group
\[
 \langle B,z\mid zez^{-1}=\alpha(e)\ (e\in E)\rangle.
\]
Here \(z\) is a new generator, called the stable letter, and
all relations of \(B\) are retained. The HNN normal form theorem
states that the natural map from \(B\) into this group is injective
\cite[Chapter~IV, Section~2]{LS}.

\begin{proposition}\label{prop:central-obstruction}
Let \(D,w\) be as in Proposition~\ref{prop:double-obstruction},
and let \(S_D\) be a finite generating list of \(D\).
Adjoin two generators \(J,z\) and define
\begin{equation}\label{eq:lambda-presentation}
 \Lambda=
 \left\langle D,J,z\ \middle|\
 \begin{array}{l}
 J^2=1,\quad [J,g]=1\ (g\in S_D),\\
 {}[J,z]=1,\quad zwz^{-1}=wJ
 \end{array}\right\rangle .
\end{equation}
Here \(D\) abbreviates a chosen finite presentation.
Then \(\Lambda\) is finitely presented, \(J\) is a nontrivial
central involution, and every homomorphism
\(\Lambda\to\U(\M)\) into the unitary group of a tracial matrix
ultraproduct sends \(J\) to the identity.
\end{proposition}

\begin{proof}
Let \(C_2=\Z/2\Z\) and \(B=D\times C_2\).
Identify \(D\) with \(D\times\{0\}\) and put \(J=(1,1)\).
Then \(J\) has order two and commutes with \(D\).
Since \(w\) has infinite order, the map
\[
 \Z\times C_2\longrightarrow E:=\langle w,J\rangle\le B,
 \qquad (k,\epsilon)\longmapsto w^kJ^\epsilon
\]
is an isomorphism. Indeed, it is surjective, and
\(w^kJ^\epsilon=1\) forces \(w^k=1\) and \(\epsilon=0\)
in the two coordinates, hence \(k=0\).

The assignments \(\alpha(w)=wJ\) and \(\alpha(J)=J\)
correspond to
\[
 (k,\epsilon)\longmapsto(k,\epsilon+k),
\]
with the second coordinate computed modulo \(2\).
This map preserves addition and its square is the identity,
so \(\alpha\) is an automorphism of \(E\).
Form the HNN extension of \(B\) with this automorphism and
stable letter \(z\). Its conjugation relations reduce to
\(zwz^{-1}=wJ\) and \(zJz^{-1}=J\): for every \(k\in\Z\)
and \(\epsilon\in C_2\), these imply
\[
 z(w^kJ^\epsilon)z^{-1}
 =(wJ)^kJ^\epsilon
 =w^kJ^{k+\epsilon}
 =\alpha(w^kJ^\epsilon).
\]
This calculation is valid for negative \(k\) because \(w\)
and \(J\) commute. Since \(S_D\) generates \(D\), its commutation
relations with \(J\), together with \(J^2=1\), present
\(B=D\times C_2\) over \(D\). Thus the HNN presentation is
exactly \eqref{eq:lambda-presentation}.

The normal form theorem embeds \(B\) in \(\Lambda\), so \(J\)
remains nontrivial and has order two. The defining relations make
it commute with \(D\) and \(z\), so it is central.
The presentation is finite because \(D\) is finitely presented
and only finitely many generators and relations were added.

Finally, if \(\pi:\Lambda\to\U(\M)\) is a homomorphism,
its restriction to \(D\) sends \(w\) to the identity by
Proposition~\ref{prop:double-obstruction}. Applying \(\pi\)
to \(zwz^{-1}=wJ\) gives
\[
 1=\pi(z)\pi(w)\pi(z)^{-1}=\pi(w)\pi(J)=\pi(J).\qedhere
\]
\end{proof}

\subsection{The specified instance}

We now take \(Q,H\) from Section~\ref{sec:cover}, with
\[
 t=s_{12},\qquad h=X_{12}(2),\qquad w=[h,t_2t_1^{-1}].
\]
Propositions~\ref{prop:kazhdan-pair} and~\ref{prop:strict-compression}
verify all hypotheses above. In the two copies of the presentation
of \(Q\), identify the \(24\) root generators with nonnegative
coefficient labels, which generate \(H\). The \(18\) generators
with negative labels and the six shear generators remain separate
in the two factors. Thus the specified presentation of \(D\)
has \(24+2(18+6)=72\) generators. Adjoining \(J,z\) gives
\(74\) generators for \(\Lambda\).

Retain both copies of every defining relation of \(Q\), with
shared generator names identified, including any repeated relators.
The universal property used in Proposition~\ref{prop:double-obstruction}
shows that this presents \(D\). Add the relations in
\eqref{eq:lambda-presentation}, taking \(S_D\) to be the full
list of \(72\) generators. This is the finite presentation used
in Section~\ref{sec:translation}.

\section{An explicit binary linear system}
\label{sec:translation}

We construct a binary linear system $A\xi=b$ whose solution group
$\Sigma=\Gamma(A,b)$ contains the group $\Lambda$ from
Section~\ref{sec:obstruction}, with $J$ mapped to $J_\Sigma$.
This preserves the nontrivial central involution needed in
Section~\ref{sec:games}. We also determine the dimensions of $A$
and prove that every row has three nonzero entries, every column
has even sum, and $b$ has a single nonzero entry. In particular,
the classical system is inconsistent.

The construction uses two embedding theorems of Slofstra~\cite{Slofstra}.
We first replace the generators of $\Lambda$ by words in involutions,
then apply the wagon wheel construction to obtain the linear system.
We specify the relation list and all row and column indices, so the
formulas below determine $(A,b)$ completely.

\subsection{The numbered presentation and its canonical relation list}

Order the six pairs of distinct elements of $\{1,2,3\}$ as
\[
 (\alpha_1,\ldots,\alpha_6)
   =((1,2),(1,3),(2,1),(2,3),(3,1),(3,2)).
\]
We write $X_{\alpha_p}(c)=X_{ij}(c)$ and $s_{\alpha_p}=s_{ij}$
when $\alpha_p=(i,j)$. A superscript $^{(f)}$ denotes the copy
in factor $f\in\{1,2\}$ of $D=Q*_H Q$. The generators with
coefficient labels $0,1,2,3$ are shared by the two factors;
the remaining generators have separate copies.
Number the generating list $S_\Lambda=(g_1,\ldots,g_{74})$ by
\begin{align}
 g_{4(p-1)+c+1}&=X_{\alpha_p}(c)
       &&(1\le p\le6,\ 0\le c\le3), \label{eq:lambda-positive-numbering}\\
 g_{24+24(f-1)+3(p-1)+c}&=X_{\alpha_p}(-c)^{(f)}
       &&(f\in\{1,2\},\ 1\le p\le6,\ 1\le c\le3),\notag\\
 g_{42+24(f-1)+p}&=s_{\alpha_p}^{(f)}
       &&(f\in\{1,2\},\ 1\le p\le6),
       \label{eq:lambda-shear-numbering}\\
 g_{73}&=J,\qquad g_{74}=z. \label{eq:lambda-last-numbering}
\end{align}
Thus positions $1$ to $24$ are shared, positions $25$ to $48$
belong to the first factor alone, and positions $49$ to $72$
belong to the second factor alone. In particular,
\[
 h=g_3,\qquad t_1=g_{43},\qquad t_2=g_{67},\qquad
 w=g_3g_{67}g_{43}^{-1}g_3^{-1}g_{43}g_{67}^{-1}.
\]
The negative coefficient in $X_{\alpha_p}(-c)^{(f)}$ is part
of the generator's name; it does not denote an inverse.

Let $R_{\mathrm{raw}}$ be the list consisting of both complete copies
of the defining relations of $Q$, translated by
\eqref{eq:lambda-positive-numbering}--\eqref{eq:lambda-shear-numbering},
followed by
\begin{equation}
 J^2,\qquad [J,g_k]\ (1\le k\le72),\qquad [J,z],\qquad
 zwz^{-1}(wJ)^{-1}.
 \label{eq:lambda-raw-relations}
\end{equation}
Each equation $u=v$ is recorded as the word $uv^{-1}$, before
cancellation. Repeated and trivial words are retained at this stage;
their lengths count all signed generator occurrences.
As shown in Section~\ref{sec:obstruction}, identifying the generators
of $H$ in the two copies presents $Q*_H Q$, and adding
\eqref{eq:lambda-raw-relations} gives exactly
\eqref{eq:lambda-presentation}.

Each copy of $Q$ contributes $7\,729$ words of total length $68\,712$.
The additional list has $75$ words of total length
$2+72\cdot4+4+15=309$, since $w$ has length $6$.
Thus $R_{\mathrm{raw}}$ contains
\begin{equation}
 15\,533\ \text{relation words of total length}\ 137\,733.
 \label{eq:raw-lambda-counts}
\end{equation}

We now normalize the relation list. Represent each word by a tuple
of elements of $\{-74,\ldots,-1,1,\ldots,74\}$, where $k$ denotes $g_k$ and
$-k$ denotes $g_k^{-1}$. A tuple is \emph{freely reduced} when it
has no adjacent inverse letters, and \emph{cyclically reduced} when
it is freely reduced and, when nonempty, its first and last letters
are not inverse.
The empty tuple represents the identity. For a reduced word $r$,
write $|r|$ for its length. Normalize the raw list as follows:
\begin{enumerate}
 \item Scan each tuple from left to right with an initially empty stack.
       If the next letter is the inverse of the top letter, remove the
       top letter; otherwise append the next letter.
 \item Repeatedly cancel inverse first and last letters. Discard the
       tuple if it becomes empty.
 \item Choose the lexicographically least tuple among all cyclic
       rotations of the remaining tuple and its inverse.
 \item Remove repetitions and sort the resulting tuples lexicographically.
\end{enumerate}
The inverse of $(a_1,\ldots,a_\ell)$ is $(-a_\ell,\ldots,-a_1)$,
and the cyclic rotations of the original tuple are
$(a_{k+1},\ldots,a_\ell,a_1,\ldots,a_k)$ for $0\le k<\ell$.
Lexicographic order uses the ordinary order on signed integers;
a proper prefix precedes the longer tuple. Denote the resulting
ordered list by $R_0=(r_1,\ldots,r_{N_0})$.

The following lemma shows that normalization preserves the specified
group and the maximal relation defect. We use the defect notation
from Section~\ref{sec:preliminaries}.

\begin{lemma}
\label{lem:canonical-lambda}
The presentation $\langle S_\Lambda\mid R_0\rangle$ presents $\Lambda$
with the specified images of its generators. For every $d\ge1$ and
every homomorphism $f:F(S_\Lambda)\to\U(d)$,
\[
 \operatorname{def}_{R_{\mathrm{raw}}}(f)
 =\operatorname{def}_{R_0}(f).
\]
Every $r_i$ is nonempty and cyclically reduced, and
\begin{equation}
 N_0=12\,863,\qquad \sum_{i=1}^{N_0}|r_i|=118\,095.
 \label{eq:canonical-lambda-counts}
\end{equation}
\end{lemma}

\begin{proof}
After each scanned prefix, the stack is reduced and represents that
prefix: appending preserves reduction unless the last pair cancels,
in which case popping performs that cancellation. Thus the first step
preserves the element of $F(S_\Lambda)$. Removing the end letters of
$aua^{-1}$ replaces it by a conjugate, as does rotating $uv$ to
$vu=u^{-1}(uv)u$. Conjugation and inversion preserve the normal
closure of a relator. Deleting trivial words or repetitions and
reordering also preserve the collective normal closure. Hence the
two presentations have the same kernel in $F(S_\Lambda)$, so their
identification fixes every generator.

For a unitary matrix $U$, the normalized Hilbert--Schmidt distance
$\norm{U-I_d}_{2,d}$ is unchanged by unitary conjugation or inversion.
The same operations therefore preserve each relation error, and
deleting repetitions or zero errors preserves their maximum.
This proves the defect identity.

After cancellation, the remaining tuple is freely reduced and has
noninverse end letters. Rotation and inversion preserve these cyclic
adjacency conditions, so every retained word is cyclically reduced.
The counts in~\eqref{eq:canonical-lambda-counts} follow by the specified
finite enumeration; Appendix~\ref{app:data} gives its complete length
histogram and reproducibility details. Summing the frequencies gives
$N_0$, and summing length times frequency gives the total length.
The complete list is defined by \lean{ThomGame.Lambda.normalizedRelators}
in the accompanying Lean repository, where its equality with the
checked table and both counts are proved.
\end{proof}

In particular, the square relation for $J$ survives normalization
as the tuple $(-73,-73)$, representing $J^{-2}$.

\subsection{A presentation by involutions}

A \emph{group over $\Z/2\Z$} has a specified central element whose
square is the identity; that element may be trivial. A homomorphism
over $\Z/2\Z$ preserves the specified element, and an embedding
is an injective such homomorphism \cite[Definition~4.1]{Slofstra}.
For a finite set $Y$, let
\[
 F_2(Y)=\langle Y\mid y^2=1\ (y\in Y)\rangle
\]
be the free product of copies of $\Z/2\Z$ indexed by $Y$.
This group notation is distinct from the field $\F_2$.
A presentation by involutions
$\operatorname{Inv}\langle Y\mid R\rangle$ means the quotient of
$F_2(Y)\times\langle J_*\mid J_*^2=1\rangle$ by the listed relators
$R$. Thus the square relations and $[y,J_*]=1$ are implicit.
The letters of $Y$ are called \emph{ordinary}, to distinguish them
from the central letter $J_*$.

Before imposing $R$, every element has a unique reduced form
\[
 J_*^p z_0\cdots z_{\ell-1},\qquad
 p\in\{0,1\},\quad z_j\in Y,
\]
with consecutive ordinary letters distinct. Its \emph{ordinary length}
is $\ell$. For positive ordinary length, it is \emph{cyclically reduced}
when also $z_{\ell-1}\ne z_0$. In this convention, words of ordinary
length one are not cyclically reduced; words of length zero are
$1$ and $J_*$.

We apply \cite[Proposition~4.3]{Slofstra} to
$\Lambda=\langle S_\Lambda\mid R_0\rangle$, using the representative
$g_{73}$ of $J$ and taking every substitution exponent to be $2$.
The theorem allows original generators of finite order.
Put $T=148$ and introduce ordinary involutions
\[
 Y=\{y_1,\ldots,y_T\},\qquad
 u_k=y_{2k-1},\quad v_k=y_{2k}\quad(1\le k\le74).
\]
Define a homomorphism
$\phi:F(S_\Lambda)\to F_2(Y)\times\langle J_*\rangle$ by
\begin{equation}
 \phi(g_k)=y_{2k-1}y_{2k}y_{2k-1}y_{2k},\qquad
 \phi(g_k^{-1})=y_{2k}y_{2k-1}y_{2k}y_{2k-1}.
 \label{eq:involution-substitution}
\end{equation}
Let $\mathcal K$ be the group defined by the ordered presentation
\begin{equation}
 \mathcal I=\operatorname{Inv}\langle
 Y\mid \phi(r_1),\ldots,\phi(r_{N_0}),J_*\phi(g_{73})\rangle.
 \label{eq:involution-presentation}
\end{equation}
Here $\mathcal I$ denotes the presentation and $\mathcal K$ its group.
The original square relation becomes $(v_{73}u_{73})^4=1$.
The appended relation
\[
 J_*y_{145}y_{146}y_{145}y_{146}=1
\]
instead identifies $\phi(g_{73})$ with $J_*$.
By Lemma~\ref{lem:canonical-lambda} and
Proposition~\ref{prop:central-obstruction}, $g_{73}$ represents a
central involution in the source, as required by Slofstra's theorem.
It therefore gives an embedding
\begin{equation}
 \iota_0:\Lambda\longrightarrow\mathcal K,
 \qquad \iota_0(g_k)=(y_{2k-1}y_{2k})^2,
 \qquad \iota_0(J)=J_*.
 \label{eq:first-embedding}
\end{equation}

To apply the wagon wheel theorem, we next verify the condition called
\emph{collegiality} in \cite[Definition~4.6]{Slofstra}.
For a reduced word $R=J_*^p z_0\cdots z_{\ell-1}$, define
\[
 \operatorname{mult}(y;R)=\#\{j:0\le j<\ell,\ z_j=y\}.
\]
Two distinct ordinary generators are \emph{adjacent} if they appear consecutively
in the cyclic ordinary word, including its last and first positions.
A finite presentation by involutions is \emph{collegial} if its explicit
relators are cyclically reduced, exclude $1$, $J_*$, and single ordinary
generators, and satisfy this condition: if a generator has odd
multiplicity in some relator, every generator adjacent to it in any
relator has even multiplicity in every relator.
For our presentation, all ordinary multiplicities will be even.

\begin{lemma}
\label{lem:collegial-input}
The presentation $\mathcal I$ is collegial. Its explicit relators,
in the displayed order and without further rotation, have reduced forms
\begin{equation}
 R_i=J_*^{p_i}y_{s_{i,0}}\cdots y_{s_{i,n_i-1}}
       \quad(1\le i\le N),\qquad N=N_0+1=12\,864.
 \label{eq:wheel-input-words}
\end{equation}
Here $n_i$ is the ordinary length, $p_i\in\{0,1\}$ is the central
exponent, and $s_{i,j}\in\{1,\ldots,T\}$ is the index of the ordinary
letter at position $j$, starting at $0$. These integers $s_{i,j}$
are distinct in meaning from the shear generators $s_{ab}$ of $Q$.
Every $n_i\ge4$, every ordinary multiplicity is even, and
\begin{equation}
 p_i=0\ (i<N),\quad p_N=1,\quad
 (s_{N,0},s_{N,1},s_{N,2},s_{N,3})=(145,146,145,146).
 \label{eq:last-wheel-word}
\end{equation}
Moreover, the total ordinary length is
\begin{equation}
 L:=\sum_{i=1}^{N}n_i
      =4\sum_{i=1}^{N_0}|r_i|+4=472\,384.
 \label{eq:total-wheel-length}
\end{equation}
\end{lemma}

\begin{proof}
Each block in~\eqref{eq:involution-substitution} alternates between
two distinct letters. Blocks from different original generators use
disjoint pairs, so their boundary letters differ. For the same generator,
a positive block begins with $u_k$ and ends with $v_k$, whereas a
negative block begins with $v_k$ and ends with $u_k$. Equal boundary
letters therefore occur exactly for opposite signs, which would give
an adjacent inverse pair in $r_i$. Free reduction excludes this within
the word, and cyclic reduction excludes it at the last and first
blocks. Thus $\phi(r_i)$ is cyclically reduced as displayed and has
ordinary length $4|r_i|$.

Each original signed occurrence contributes twice to each of its two
ordinary letters, so every multiplicity is even. The appended relator
is also cyclically reduced, has ordinary length four, and contains
each of $y_{145},y_{146}$ twice. Every relator therefore has length
at least four, and the adjacency condition is vacuous. This proves
collegiality and~\eqref{eq:last-wheel-word}; summing lengths gives
\eqref{eq:total-wheel-length}.

Finally, parsing a substituted word into successive blocks of length
four recovers its original signed tuple. Distinct $r_i$ therefore give
distinct reduced words. The appended relator has central exponent one,
so differs from all of them. Hence there are exactly $N_0+1$ explicit
relators. This is the even exponent argument used in the proof of
\cite[Corollary~4.8]{Slofstra}.
\end{proof}

\subsection{The rows and columns of the matrix}

An \emph{occurrence} $(i,j)$ specifies position $j$ in relation $R_i$.
Let
\[
 o_i=\sum_{a=1}^{i-1}n_a\quad(1\le i\le N),\qquad
 m=3L,\qquad n=4L+T.
\]
The offset $o_i$ counts positions in the preceding words.
Positions $j=0,\ldots,n_i-1$ are read cyclically: $j+1$ and $j-1$
in subscripts are reduced modulo $n_i$. Offsets and matrix indices
are ordinary integers.

The first $T$ columns correspond to $y_1,\ldots,y_T$.
For each occurrence $(i,j)$, put $q=o_i+j$ and assign four auxiliary
column indices
\begin{equation}
 \begin{split}
 a_{i,j}&=T+4q+1,\qquad b_{i,j}=T+4q+2,\\
 c_{i,j}&=T+4q+3,\qquad d_{i,j}=T+4q+4.
 \end{split}
 \label{eq:auxiliary-column-indices}
\end{equation}
In particular, $b_{i,j}$ is an integer index, distinct from the vector
$b$. Write $\xi=(\xi_1,\ldots,\xi_n)^{\mathsf T}\in\F_2^n$ for
the unknown column vector and $b_r$ for coordinate $r$ of $b\in\F_2^m$.
Here $^{\mathsf T}$ denotes transpose.
Define the following rows over $\F_2$, in this order:
\begin{equation}
 \begin{aligned}
 \xi_{s_{i,j}}+\xi_{a_{i,j}}+\xi_{b_{i,j}}
       &=p_i\mathbf 1_{\{j=0\}},
       &&\text{row }3(o_i+j)+1,\\
 \xi_{b_{i,j}}+\xi_{c_{i,j}}+\xi_{a_{i,j+1}}
       &=0,
       &&\text{row }3(o_i+j)+2,\\
 \xi_{c_{i,j}}+\xi_{d_{i,j}}+\xi_{d_{i,j+1}}
       &=0,
       &&\text{row }3(o_i+j)+3.
 \end{aligned}
 \label{eq:wheel-rows}
\end{equation}
The indicator $\mathbf 1_{\{j=0\}}$ is $1$ at $j=0$ and $0$ elsewhere.
Each listed variable has coefficient $1$; all other entries are zero.
These rules define $A\in\Mat_{m,n}(\F_2)$ and $b$ completely.

Indeed, the intervals $\{o_i,\ldots,o_i+n_i-1\}$ partition
$\{0,\ldots,L-1\}$. Hence the occurrences index all $3L$ rows and
all $4L$ auxiliary columns without omissions or repetitions.
Different auxiliary types $a,b,c,d$ have distinct residues modulo $4$
after subtracting $T$, and none is an ordinary column.
The first two row types therefore contain three distinct columns.
In the third, the $c$ column differs from both $d$ columns, and
$d_{i,j}\ne d_{i,j+1}$ because $n_i\ge4$. Thus every row has exactly
three nonzero entries, with no cancellation in $\F_2$.

\subsection{The solution group and preservation of \texorpdfstring{$J$}{J}}

To apply Slofstra's theorem, we identify the matrix with the incidence
matrix of the wagon wheel hypergraph. Here a \emph{hypergraph} consists
of vertices, edges, and their incidences; an edge may meet more than
two vertices. Regard rows as vertices and columns as edges.
Wheel $i$ has vertex set
\[
 \mathcal V_i=\{(i,j,k):0\le j<n_i,\ k\in\{1,2,3\}\},
\]
with $(i,j,k)$ corresponding to row $3(o_i+j)+k$.
The full vertex set is the disjoint union
$\mathcal V=\bigcup_{i=1}^N\mathcal V_i$,
and the edge set is $\{1,\ldots,n\}$.
An ordinary edge $e\le T$ meets $(i,j,1)$ exactly when $s_{i,j}=e$,
and has no other incidences. The auxiliary incidences are
\begin{align*}
 a_{i,j}&:\ (i,j-1,2),\ (i,j,1),&
 b_{i,j}&:\ (i,j,1),\ (i,j,2),\\
 c_{i,j}&:\ (i,j,2),\ (i,j,3),&
 d_{i,j}&:\ (i,j-1,3),\ (i,j,3).
\end{align*}
Reading off the edges at each vertex gives exactly the three row
supports in~\eqref{eq:wheel-rows}. These are the incidences of the
wagon wheel in \cite[Section~5]{Slofstra}.

A \emph{vertex labelling} assigns an element of $\F_2$ to each vertex.
It is an \emph{$\mathcal I$-labelling} if its sum on $\mathcal V_i$
is $p_i$. Our vector $b$ gives such a labelling: on wheel $i$,
only $(i,0,1)$ can have a nonzero label, and that label is $p_i$.

Using the conventions of Section~\ref{subsec:statement-language},
the solution group $\Sigma=\Gamma(A,b)$ has generators
$x_1,\ldots,x_n,J_\Sigma$ and relations
\begin{equation}
 \begin{gathered}
 x_e^2=1\quad(1\le e\le n),\qquad J_\Sigma^2=1,\qquad
 [x_e,J_\Sigma]=1\quad(1\le e\le n),\\
 {}[x_e,x_f]=1\quad\text{if $e<f$ belong to the support of one row},\\
 x_{e_1}x_{e_2}x_{e_3}=J_\Sigma^{b_r}
 \quad\text{if row $r$ has support $e_1<e_2<e_3$}.
 \end{gathered}
 \label{eq:explicit-solution-presentation}
\end{equation}
The support is the set of columns with nonzero entries. Each unordered
pair of distinct columns occurring together in a row gives one
commutation relation, even if it occurs in several rows. We use increasing order in each row product;
the commutation relations make its value independent of this order.

For a collegial presentation, \cite[Theorem~5.1]{Slofstra} embeds its
group in the solution group of its wagon wheel, sending each ordinary
generator to its edge generator and preserving the central involution.
Definition~5.2 and Lemma~5.3 of that paper allow any
$\mathcal I$-labelling. To see why this includes our specified $b$,
observe that the auxiliary graph on each wheel is connected: the
$d$ edges form a cycle through the vertices with third coordinate $3$,
and the $c$ and $b$ edges connect the other vertices to that cycle.
Two admissible labellings differ on an even number of vertices in each
wheel. Pair these vertices and join each pair by an auxiliary path.
The sum of its incidence columns toggles just the two endpoints,
since every internal vertex occurs twice.

Changing a labelling by one auxiliary column is implemented by
$x_e\mapsto J_\Sigma x_e$ for that column, fixing all other generators.
Centrality and $J_\Sigma^2=1$ preserve the square and commutation
relations, while each affected row product gains precisely the factor
$J_\Sigma$ required by its changed label. The substitution is its own
inverse. Composing these isomorphisms along the paths therefore realizes
any admissible change of labelling while fixing every ordinary edge
generator and $J_\Sigma$.

The two embedding theorems now give the required map from $\Lambda$
to the solution group of our specific matrix.

\begin{proposition}
\label{prop:solution-embedding}
There is an injective homomorphism
\[
 \iota:\Lambda\longrightarrow\Sigma,\qquad
 \iota(g_k)=(x_{2k-1}x_{2k})^2\quad(1\le k\le74),
 \qquad \iota(J)=J_\Sigma.
\]
In particular, $J_\Sigma\ne1$.
\end{proposition}

\begin{proof}
Lemma~\ref{lem:collegial-input} verifies collegiality, and the incidence
and labelling checks above identify $(A,b)$ with an allowed wagon wheel.
Slofstra's theorem therefore gives an embedding
\[
 \iota_1:\mathcal K\longrightarrow\Sigma,
 \qquad y_e\longmapsto x_e\ (1\le e\le T),
 \qquad J_*\longmapsto J_\Sigma.
\]
Compose it with $\iota_0$ from~\eqref{eq:first-embedding}.
The composite is injective and has the stated images.
Since $J\ne1$ by Proposition~\ref{prop:central-obstruction},
its image $J_\Sigma$ is nontrivial.
\end{proof}

\subsection{Dimensions and the parity obstruction}

The word lengths determine the dimensions of the system. Its column
parities also give a direct proof of classical inconsistency.
Write $e_r\in\F_2^m$ for the standard basis vector with its sole
nonzero coordinate at $r$. A column has \emph{even sum} if its number
of nonzero entries is even, equivalently if its sum in $\F_2$ is zero.

\begin{proposition}
\label{prop:matrix-size}
The construction gives
\[
 A\in\Mat_{1\,417\,152,\;1\,889\,684}(\F_2),\qquad
 b=e_{1\,417\,141},\qquad
 \#\{(r,e):A_{r,e}=1\}=4\,251\,456.
\]
Every row has exactly three nonzero entries and every column has even
sum. Consequently $A\xi=b$ has no solution over $\F_2$.
\end{proposition}

\begin{proof}
By~\eqref{eq:total-wheel-length}, there are $L=472\,384$ ordinary
occurrences and $T=148$ ordinary columns. The row construction gives
\[
 m=3L=1\,417\,152,\qquad n=4L+T=1\,889\,684.
\]
Each row has three distinct nonzero entries, so their total number is
$3m=4\,251\,456$. By~\eqref{eq:last-wheel-word}, only the last wheel
has nonzero central exponent, and its length is four.
Thus $o_N=L-4$, and~\eqref{eq:wheel-rows} gives the unique nonzero
coordinate of $b$ at
\[
 3o_N+1=3(L-4)+1=1\,417\,141.
\]

For an ordinary column $e\le T$, each occurrence of $y_e$ contributes
one entry in a distinct first row. Its integer column sum is therefore
$\sum_{i=1}^N\operatorname{mult}(y_e;R_i)$, which is even by
Lemma~\ref{lem:collegial-input}. Each auxiliary column has exactly
two incident vertices, as listed above. The endpoints are distinct:
for $a,b,c$ they have different third coordinates, and for $d$ they
have different cyclic positions since $n_i\ge4$. These columns also
have even sum.

Let $\boldsymbol{1}_m=(1,\ldots,1)^{\mathsf T}\in\F_2^m$.
The column parities and the single nonzero entry of $b$ imply
\[
 \boldsymbol{1}_m^{\mathsf T}A=0,\qquad
 \boldsymbol{1}_m^{\mathsf T}b=1.
\]
If $A\xi=b$, multiplication by $\boldsymbol{1}_m^{\mathsf T}$ would
give $0=1$ in $\F_2$, a contradiction.
\end{proof}

\section{Approximation and the values of the game}\label{sec:games}

We now prove that the game constructed in Section~\ref{sec:translation}
separates $C_{qa}$ from $C_{qc}$. The distinguished involution
$J_\Sigma$ is nontrivial in the solution group, which gives a perfect
commuting strategy. Its image nevertheless tends uniformly to the
identity in approximate matrix representations, forcing the quantum
value below one. We then compute the classical value exactly and
complete the proofs of the main theorems.

\subsection{The question distribution}

We retain the game $G=G(A,b)$ and the strategy models from
Section~\ref{subsec:statement-language}. Recall that
$V_r=\{e_{r,1}<e_{r,2}<e_{r,3}\}$ is the support of row $r$.
Let
\[
 \mathcal E=\{(r,e):A_{re}=1\},\qquad
 |\mathcal E|=3m=4\,251\,456,
\]
and write $\operatorname{pos}_r(e)=j$ when $e=e_{r,j}$.
For a correlation $p$, its success probability is
\begin{equation}\label{eq:success-functional}
 \omega(G,p)
 =\frac1{3m}\sum_{(r,e)\in\mathcal E}
   \sum_{\substack{a\in\F_2^3,\,\beta\in\F_2\\
          a_1+a_2+a_3=b_r\\
          a_{\operatorname{pos}_r(e)}=\beta}}
       p(a,\beta\mid r,e).
\end{equation}
The parity equality in the inner sum is in $\F_2$.

The approximation criterion used below is stated for uniform sampling
of all $mn$ row and column pairs. We therefore compare that convention
with our uniform incidence distribution. Call an Alice answer
\emph{legal} if it has the required row parity. For each row, fix one
legal answer and replace every illegal answer by it. This local change
cannot lower incidence success, because an illegal answer always loses.
It preserves every strategy model: for POVMs, sum the effects assigned
to the same new answer. The operation is linear and continuous on
correlation tables, so it also preserves $C_{qa}$.

For a strategy with legal Alice answers, let $P_{\rm inc}$ be its
success probability in our game and $P_{\rm all}$ its success probability
when all row and column pairs are sampled uniformly, with pairs outside
$\mathcal E$ accepted automatically. If $\ell_{r,e}$ is its conditional
losing probability at incidence $(r,e)$, then
\[
 1-P_{\rm inc}=\frac1{3m}\sum_{(r,e)\in\mathcal E}\ell_{r,e},
 \qquad
 1-P_{\rm all}=\frac1{mn}\sum_{(r,e)\in\mathcal E}\ell_{r,e}.
\]
Consequently
\begin{equation}\label{eq:distribution-comparison}
 1-P_{\rm all}=\frac3n(1-P_{\rm inc}).
\end{equation}
Legal answers also suffice for the convention sampling all pairs.
Since $n$ is fixed, a sequence approaches success one in one convention
exactly when it does in the other. Perfect strategies likewise coincide.

In the quantum model, we allow arbitrary POVMs on finite dimensional
spaces. After the
replacement above, Alice has four legal outcomes on each row. Applying
\cite[Lemma~4.1]{SV} to these POVMs and Bob's binary POVMs gives
simultaneous projective dilations on larger finite dimensional spaces,
preserving every correlation entry and the legal outcome sets. Thus the use of projective
measurements in the cited approximation criterion does not restrict
our quantum model.

\subsection{Uniform triviality and the quantum value}

The embedding from Section~\ref{sec:translation} transfers the
obstruction for $J$ to $J_\Sigma$. Lemma~\ref{lem:ultraproduct}
then expresses this obstruction as a tolerance independent of the
matrix dimension. Let $\langle S_\Sigma\mid R_\Sigma\rangle$ be the
finite presentation in~\eqref{eq:explicit-solution-presentation},
with each equation $u=v$ recorded as $uv^{-1}$. As before,
$f(J_\Sigma)$ denotes evaluation of its distinguished generator.

\begin{proposition}\label{prop:uniform-solution}
The element $J_\Sigma$ is nontrivial. For every $\eta>0$, there is
$\delta>0$ such that, for every integer $d\ge1$ and homomorphism
$f:F(S_\Sigma)\to\U(d)$,
\[
 \defect_{R_\Sigma}(f)\le\delta
 \quad\Longrightarrow\quad
 \norm{f(J_\Sigma)-I_d}_{2,d}<\eta.
\]
\end{proposition}

\begin{proof}
Proposition~\ref{prop:solution-embedding} gives an embedding
$\iota:\Lambda\to\Sigma$ with $\iota(J)=J_\Sigma\ne1$.
For any tracial matrix ultraproduct $\M$ and homomorphism
$\theta:\Sigma\to\U(\M)$, Proposition~\ref{prop:central-obstruction}
applied to $\theta\circ\iota$ gives
\[
 \theta(J_\Sigma)=(\theta\circ\iota)(J)=1.
\]
Since this holds for every $\M$ and $\theta$,
Lemma~\ref{lem:ultraproduct} gives the stated uniform tolerance
for the word $J_\Sigma$ in the specified presentation.
\end{proof}

To rule out quantum value one, we use the following criterion of
Slofstra and Vidick \cite[Corollary~5.3]{SV}, restated using the
defect from Definition~\ref{def:approx}. It requires one fixed
positive separation from the identity while the relation errors
tend to zero.

\begin{proposition}\label{thm:SV}
For a binary linear system game with solution group
$\Gamma(A,b)=\langle S\mid R\rangle$ in its standard presentation
and distinguished generator $J$, the quantum value is one if and
only if there are $\eta>0$, positive integers $d_k$, and
homomorphisms $f_k:F(S)\to\U(d_k)$ such that
\[
 \defect_R(f_k)\longrightarrow0,
 \qquad
 \norm{f_k(J)-I_{d_k}}_{2,d_k}\ge\eta
 \quad(k\ge1).
\]
\end{proposition}

The norms in this criterion are normalized Hilbert--Schmidt norms,
and the dimensions $d_k$ need not be bounded. Equation
\eqref{eq:distribution-comparison} makes the criterion applicable
to our incidence distribution. Proposition~\ref{prop:uniform-solution}
excludes exactly the sequence it would require.

\begin{corollary}\label{cor:quantum-strict}
For the specified game $G$,
\[
 \omega_q(G)=\omega_{qa}(G)<1.
\]
\end{corollary}

\begin{proof}
If $\omega_q(G)=1$, Proposition~\ref{thm:SV} gives a sequence
$f_k$ and a fixed $\eta>0$ as above. Choose $\delta>0$ for this
$\eta$ using Proposition~\ref{prop:uniform-solution}.
For sufficiently large $k$, the defect is at most $\delta$, so
\[
 \norm{f_k(J_\Sigma)-I_{d_k}}_{2,d_k}<\eta,
\]
contrary to the criterion. Finally, the success functional is
continuous and $C_{qa}=\overline{C_q}$, so its supremum over $C_{qa}$
equals its supremum over $C_q$.
\end{proof}

\subsection{A perfect commuting strategy}

The nontriviality of $J_\Sigma$ has a different consequence in the
commuting model. The construction of Cleve, Liu, and Slofstra~\cite{CleveLiuSlofstra2017}
uses the left and right regular representations to produce perfectly
consistent measurements.
We give the construction for our game to make the strategy explicit.

\begin{proposition}\label{prop:commuting}
The game $G$ admits a perfect commuting strategy. In particular,
$\omega_{qc}(G)=1$.
\end{proposition}

\begin{proof}
On $\ell^2(\Sigma)$, with orthonormal basis
$\{\delta_g:g\in\Sigma\}$, define the left and right regular
representations by
\[
 \lambda(g)\delta_k=\delta_{gk},\qquad
 \rho(g)\delta_k=\delta_{kg^{-1}}.
\]
These are unitary representations and commute with one another,
as their actions on the basis show. Let
\[
 \HH_-=\{\xi\in\ell^2(\Sigma):\lambda(J_\Sigma)\xi=-\xi\},
 \qquad
 \Omega=\frac{\delta_1-\delta_{J_\Sigma}}{\sqrt2}.
\]
Since $J_\Sigma\ne1$ and $J_\Sigma^2=1$, the vector $\Omega$ is
a unit vector in $\HH_-$. Centrality of $J_\Sigma$ and commutation
of the two regular representations make $\HH_-$ invariant under
all $\lambda(g)$, $\rho(g)$, and their inverses.

Set $X_e=\lambda(x_e)|_{\HH_-}$ and
$Y_e=\rho(x_e)|_{\HH_-}$. They are selfadjoint unitaries because
$x_e^2=1$. For each row $r$, its three $X_e$ commute and satisfy
\[
 \prod_{j=1}^3X_{e_{r,j}}=(-1)^{b_r}I_{\HH_-}.
\]
Define Alice's and Bob's measurement effects by
\begin{equation}\label{eq:joint-projectors}
 P^r_a=\prod_{j=1}^3
       \frac{I_{\HH_-}+(-1)^{a_j}X_{e_{r,j}}}{2},
 \qquad
 Q^e_\beta=\frac{I_{\HH_-}+(-1)^\beta Y_e}{2},
\end{equation}
where $a\in\F_2^3$ and $\beta\in\F_2$.
For a selfadjoint unitary $Z$, $(I+Z)/2$ and $(I-Z)/2$ are
orthogonal projections summing to the identity. The commuting
factors in Alice's product therefore give mutually orthogonal
projections whose sum is the identity; the same holds for Bob.
All Alice effects commute with all Bob effects because left and
right regular operators commute.

By~\eqref{eq:joint-projectors},
$X_{e_{r,j}}P^r_a=(-1)^{a_j}P^r_a$. Multiplying these identities
and using the row product gives
\[
 (-1)^{a_1+a_2+a_3}P^r_a=(-1)^{b_r}P^r_a.
\]
Thus $P^r_a=0$ whenever $a$ has the wrong parity.
Moreover, centrality of $J_\Sigma$ and $x_e^{-1}=x_e$ imply
\[
 X_e\Omega
 =\frac{\delta_{x_e}-\delta_{x_eJ_\Sigma}}{\sqrt2}
 =Y_e\Omega.
\]
At incidence $e=e_{r,j}$, Alice's marginal for bit $a_j$ is
$(I+(-1)^{a_j}X_e)/2$, obtained by summing over the other bits.
The probability of disagreement with Bob is consequently
\begin{align*}
 &\left\langle\Omega,
 \left(\frac{I+X_e}{2}\frac{I-Y_e}{2}
       +\frac{I-X_e}{2}\frac{I+Y_e}{2}\right)\Omega\right\rangle\\
 &\qquad=\frac12\bigl(1-\langle\Omega,X_eY_e\Omega\rangle\bigr)
 =\frac14\norm{(X_e-Y_e)\Omega}^2=0.
\end{align*}
Here the norm is the Hilbert space norm; the second equality uses
commutation and $X_e^2=Y_e^2=I$. Both winning conditions therefore
hold at every incidence with probability one.
\end{proof}

\subsection{The classical value}

The parity obstruction also determines the classical value: every
deterministic strategy loses at least one incidence, and the single
nonzero entry of $b$ makes this bound attainable.

\begin{proposition}\label{prop:classical}
For the uniform incidence distribution,
\[
 \omega_c(G)=1-\frac1{3m}
            =1-\frac1{4\,251\,456}.
\]
\end{proposition}

\begin{proof}
Bob's answers in a deterministic strategy define an assignment
$\xi\in\F_2^n$. If every incidence were won, consistency and parity
would imply $A\xi=b$, contradicting
Proposition~\ref{prop:matrix-size}.
Thus every deterministic strategy loses at least one of the $3m$
equally likely incidences. By linearity of
\eqref{eq:success-functional}, the same upper bound $1-1/(3m)$
holds for mixtures of deterministic strategies.

To attain this bound, let Bob always answer zero. Alice answers
$(0,0,0)$ when $b_r=0$ and $(1,0,0)$ on the unique row with $b_r=1$,
using the increasing order of its columns.
Every parity condition holds, and consistency fails only when that
row and its first column are queried. Hence exactly one incidence
is lost.
\end{proof}

Together with Corollary~\ref{cor:quantum-strict},
Proposition~\ref{prop:commuting}, and
$\omega_c(G)\le\omega_q(G)$, this gives
\[
 0<\omega_{qc}(G)-\omega_q(G)
   =1-\omega_q(G)\le\frac1{4\,251\,456}.
\]

\subsection{Completion of the main theorems}

\begin{proof}[Proof of Theorem~\ref{thm:main}]
Proposition~\ref{prop:matrix-size} gives the matrix dimensions,
row sizes, and vector $b$.
Proposition~\ref{prop:uniform-solution} gives
$J_\Sigma\ne1$ and~\eqref{eq:main-quantifiers}.
Finally, Corollary~\ref{cor:quantum-strict} and
Proposition~\ref{prop:commuting} give
\[
 \omega_q(G)=\omega_{qa}(G)<1=\omega_{qc}(G).
\]
\end{proof}

\begin{proof}[Proof of Theorem~\ref{thm:existence}]
The game in Theorem~\ref{thm:main} is finite, has two players,
and has payoffs in $\{0,1\}$.
Its perfect commuting correlation has success probability one,
whereas every correlation in $C_{qa}$ has success probability
at most $\omega_{qa}(G)<1$.
Thus this correlation belongs to $C_{qc}\setminus C_{qa}$.
\end{proof}

\begin{remark}\label{rem:numerical-values}
The accompanying Lean~4 development also proves the explicit
lower bound
\[
 1-\omega_q(G)\ge 2^{-2^{50003}}.
\]
The declarations \lean{paper_classical_value_corollary} and
\lean{paper_quantum_gap_corollary} in
\lean{ThomGame/Construction/PaperValueCorollaries.lean}
formalize Proposition~\ref{prop:classical} and this quantitative
bound, respectively.
The section ``Further results'' in the repository's
\lean{README.md} provides the formal statements and references
to their proofs.
\end{remark}

\section*{AI use declaration}
Generative AI tools (GPT 6 Astra and Fable 5.1) were used to assist with mathematical exploration and proof development, the implementation and computational verification of the explicit construction, the writing and debugging of Lean 4 code, and the drafting and revision of this manuscript. The authors have reviewed the mathematical proofs and the entire Lean formalization process, including all AI-assisted contributions, and take full responsibility for the mathematical claims, computations, references, and final exposition. The accompanying formal proofs are checked by Lean’s kernel under the axioms documented in the repository. The correspondence between the formal definitions and theorem statements and the mathematical claims in this manuscript remains the authors’ responsibility.

\section*{Acknowledgements.}
    The authors want to thank Li Zhou, Sizhuo Yan, Jianting Yang, Hao Liang and Jiaqi Wang for valuable discussions. Minbo Gao and Tianshi Yu are supported by the National Key R\&D Program of China 2023YFA1009403. 
    Lihong Zhi is supported by the National Key R\&D Program of China 2023YFA1009401 and the Strategic Priority Research Program of the Chinese Academy of Sciences under Grant XDA0480501.

\bibliographystyle{alphaurl}
\bibliography{explicit_game_references}

\newpage
\appendix

\section{Canonical counts and reproducibility}
\label{app:data}

This appendix records the relation counts and word lengths used in
Section~\ref{sec:translation}, and identifies the Lean sources for
reproducing them. We first count the raw defining words for $Q$ and
$\Lambda$, then give the length distribution of the canonical list
$R_0$. The generator numbering, normalization procedure, and matrix
construction are fixed in Section~\ref{sec:translation}.

\subsection{The raw relation families}

The families (E0) to (E4) are
\eqref{eq:q-power}--\eqref{eq:q-triple}.
There are six root positions and seven coefficient labels.
Each equation $u=v$ contributes the word $uv^{-1}$, whose raw length
counts all signed generator occurrences before cancellation.
We count every allowed ordered index choice, retaining repetitions
and freely trivial words, including commutators of a generator with
itself. Thus the number of instances need not equal the number of
distinct relators.

\begin{center}
\small
\begin{tabular}{lrrr}
\toprule
Family & Instances & Length per word & Total length\\
\midrule
E0 & $42$ & $5$ & $210$\\
E1 & $882$ & $4$ & $3\,528$\\
E2 & $42$ & $5$ & $210$\\
E3 & $42$ & $5$ & $210$\\
E4 & $6\,174$ & $10$ & $61\,740$\\
Shear commutators & $18$ & $4$ & $72$\\
Shear root relations & $6$ & $5$ & $30$\\
Shear torsion relation & $1$ & $12$ & $12$\\
Action, unchanged coefficients & $360$ & $4$ & $1\,440$\\
Action, changed coefficients & $144$ & $7$ & $1\,008$\\
Negative recovery & $18$ & $14$ & $252$\\
\midrule
Total & $7\,729$ & & $68\,712$\\
\bottomrule
\end{tabular}
\end{center}

For a fixed root position $(i,j)$, let $k$ be the remaining element
of $\{1,2,3\}$. The roots $(a,b)$ satisfying
$\{i,a\}\cap\{j,b\}=\varnothing$ are exactly
$(i,j)$, $(i,k)$, and $(k,j)$.
Hence E1 has $6\cdot3\cdot7^2=882$ instances.
Family E0 has $6\cdot7=42$, and each of E2 and E3 has
$3!\cdot7=42$. For E4, there are $3!$ choices of distinct
$(i,j,k)$, three compatible outer roots, and three independent
coefficient labels, giving $3!\cdot3\cdot7^3=6\,174$.

The same compatibility count gives $6\cdot3=18$ shear commutators.
The shear root and torsion relations contribute six and one instance,
respectively. For the action relations there are $6\cdot2\cdot6=72$
choices of shear $s_{ab}$, sign, and root position.
The five unchanged coefficient labels are $0$ and $\pm c$ with
$c\in\{1,2,3\}\setminus\{b\}$; the two changed labels are $\pm b$.
This gives $72\cdot5=360$ and $72\cdot2=144$ instances.
Negative recovery contributes $6\cdot3=18$ instances.

For words of raw lengths $p$ and $q$, their commutator has raw length
$2p+2q$. This gives lengths four for E1 and ten for E4.
The root relations have length five, while the action relations have
lengths $3+1=4$ or $3+4=7$.
Each recovery word $r_c$ has length six, so its relator has length
$6+1+6+1=14$.

The raw presentation of $\Lambda$ contains two copies of the table,
with shared generator names identified, and the $75$ additional words
in~\eqref{eq:lambda-raw-relations}, of total length $309$.
It therefore has
\[
 2\cdot7\,729+75=15\,533\ \text{words},\qquad
 2\cdot68\,712+309=137\,733\ \text{letters}.
\]

\begin{samepage}
\subsection{The canonical list}

Apply the normalization from Section~\ref{sec:translation}:
freely and cyclically reduce each word, discard trivial words,
choose the lexicographically least signed tuple among its cyclic rotations and their
inverses, then remove repetitions and sort.
The resulting list $R_0$ has the following length distribution.

\begin{center}
\small
\begin{tabular}{lrrrrrrrrr}
\toprule
Length & $2$ & $4$ & $5$ & $7$ & $8$ & $10$ & $12$ & $14$ & $15$\\
Count & $1$ & $1\,153$ & $186$ & $288$ & $984$ & $10\,212$ & $2$ & $36$ & $1$\\
\bottomrule
\end{tabular}
\end{center}
\end{samepage}

Summing the counts, and then weighting them by length, gives
\[
 |R_0|=12\,863,\qquad \sum_{r\in R_0}|r|=118\,095,
\]
as stated in Lemma~\ref{lem:canonical-lambda}.
Lemma~\ref{lem:collegial-input} and Proposition~\ref{prop:matrix-size}
then give the numbers of wheel occurrences, equations, and variables.

\subsection{Reproducing the finite data}

The ordered canonical list is defined by
\lean{ThomGame.Lambda.normalizedRelators}.
Its verification is recorded in two files:
\begin{itemize}
 \item \lean{ThomGame/Certificates/LambdaNormalized.lean}
       proves that the normalization program equals the explicit
       relation table.
 \item \lean{ThomGame/Groups/LambdaChecks.lean}
       checks the number of relators and their total length using
       \texttt{decide +kernel}.
\end{itemize}

The same development specifies every wheel word, matrix entry, and
coordinate of $b$, so the binary system can be generated directly
from these finite definitions. Separate expanded matrix and
presentation files are not needed.
The formalization appendix records the source correspondence,
pinned versions, and verification commands.

\section{Formalization and source concordance}\label{app:formalization}

The tables below identify the formal statements corresponding to the
main conclusions and the group construction.
The main theorems have the same mathematical statements in the manuscript
and in Lean; intermediate entries distinguish a full statement, a
specialization to the specified groups, and the consequence used in the
formal proof. Source-version information is recorded in
Table~\ref{tab:lean-checks}. All declaration names omit the prefix
\lean{ThomGame.}; source paths beginning with \lean{Analysis/},
\lean{Construction/}, \lean{Groups/}, \lean{Quantum/}, or
\lean{Certificates/} are relative to \lean{ThomGame/}.

\subsection{Final statements and mathematical models}

\begingroup
\small
\setlength{\tabcolsep}{4pt}
\renewcommand{\arraystretch}{1.08}
\begin{longtable}{@{}p{0.29\linewidth}p{0.68\linewidth}@{}}
\caption{Entry points for the main theorems.}
\label{tab:lean-conclusions}\\
\toprule Paper statement & Lean declaration and scope\\ \midrule
\endfirsthead
\toprule Paper statement & Lean declaration and scope (continued)\\ \midrule
\endhead
\bottomrule\endfoot
Theorem~\ref{thm:existence}: existence of a finite binary-payoff game &
\lean{Quantum.exists_finiteGame_povm_quantum_commuting_separation}
in \lean{Construction/NonlocalGameSeparation.lean}.
Quantifies over the four finite alphabets and the game; proves binary
payoffs and $\omega_q<1=\omega_{qc}$ for arbitrary POVMs.
Nonempty alphabets are checked by
\lean{SeparationCheck.binary_separation} in
\lean{scripts/CheckSeparation.lean}.
The equality $\omega_q=\omega_{qa}$ follows from
\lean{Quantum.FiniteGame.omegaQPOVM_eq_omegaQaPOVM}
in \lean{Quantum/POVMGameValue.lean};
the inequality $C_{qa}\ne C_{qc}$ follows from the value separation.
\\[5pt]
Theorem~\ref{thm:main}: the specified construction &
\lean{Construction.paper_main_results_povm}
in \lean{Construction/PaperPOVMGap.lean} proves $J_\Sigma\ne1$,
\eqref{eq:main-quantifiers}, and
$\omega_q=\omega_{qa}<1=\omega_{qc}$.
The matrix dimensions, row supports, and $b$ are supplied separately
by the declarations in Table~\ref{tab:lean-construction}.
\end{longtable}
These entry theorems have no unproven mathematical hypotheses.
The following definitions fix the meaning of their statements.


\begin{longtable}{@{}p{0.29\linewidth}p{0.68\linewidth}@{}}
\caption{Definitions and conventions used by the entry theorems.}
\label{tab:lean-models}\\
\toprule Paper object & Lean definition or identification\\ \midrule
\endfirsthead
\toprule Paper object & Lean definition or identification (continued)\\ \midrule
\endhead
\bottomrule\endfoot
Finite game and success functional
(Section~\ref{subsec:statement-language}) &
\lean{Quantum.FiniteGame} and \lean{Quantum.FiniteGame.success}
in \lean{Quantum/FiniteGame.lean};
\lean{Quantum.CorrelationTable} in \lean{Quantum/Correlation.lean}.
The score is the finite sum of question weight, payoff, and
conditional answer probability.\\[5pt]
POVMs and tensor-product strategies &
\lean{Quantum.POVM} in \lean{Quantum/POVM.lean};
\lean{Quantum.FinitePOVMStrategy} and its \lean{correlation}
in \lean{Quantum/FinitePOVMStrategy.lean}.
Local dimensions are arbitrary positive integers.
\lean{Quantum.FinitePOVMStrategy.exists_coordinates}, in
\lean{Quantum/FinitePOVMCoordinates.lean}, identifies the coordinate
model with arbitrary finite-dimensional local Hilbert spaces.\\[5pt]
Commuting strategies &
\lean{Quantum.CommutingPOVMStrategy} and its \lean{correlation}
in \lean{Quantum/CommutingPOVMStrategy.lean}.
The Hilbert space may be infinite-dimensional; effects commute
across the two players. Both strategy models use unit-vector states.\\[5pt]
Correlation sets and values &
\lean{Quantum.FiniteGame.value} in \lean{Quantum/GameValue.lean};
\lean{Quantum.FiniteGame.omegaQPOVM},
\lean{Quantum.FiniteGame.omegaQaPOVM}, and
\lean{Quantum.FiniteGame.omegaQcPOVM}
in \lean{Quantum/POVMGameValue.lean}.
Values are suprema of scores over the corresponding correlation
sets; the $qa$ set is the closure of the $q$ set.\\[5pt]
Indices and binary scalars &
\lean{Construction.rowEquiv_number} and
\lean{Construction.colEquiv_number}
in \lean{Construction/Numbering.lean}.
Lean indices have type \lean{Fin}; their values plus one are the
paper's row and column numbers. Binary scalars are \lean{ZMod 2}.\\[5pt]
$\Sigma$, $R_\Sigma$, and increasing row order &
\lean{SolutionGroup.Paper.relators} in
\lean{Groups/PaperSolutionGroup.lean};
\lean{Construction.paperSystem} in
\lean{Construction/PaperOrderedSystem.lean}.
These use $[x,J]$, row pairs $i<j$, and the row product in increasing
column order, as in the manuscript.\\[5pt]
$G(A,b)$ and \eqref{eq:success-functional} &
\lean{Construction.paperGame} in \lean{Construction/PaperGame.lean},
using \lean{SparseSystem.incidenceGame} in
\lean{Quantum/IncidenceGame.lean}.
The distribution is uniform on the $4{,}251{,}456$ incidences;
Alice has eight answers and Bob has two.\\
\end{longtable}

\subsection{Approximation, normalization, and the group obstruction}

\begin{longtable}{@{}p{0.29\linewidth}p{0.68\linewidth}@{}}
\caption{Analytic statements and their scope in the formal proof.}
\label{tab:lean-analysis}\\
\toprule Paper statement & Lean declarations and coverage\\ \midrule
\endfirsthead
\toprule Paper statement & Lean declarations and coverage (continued)\\ \midrule
\endhead
\bottomrule\endfoot
Normalized Hilbert--Schmidt norm &
\lean{Analysis.hsNorm} and \lean{Analysis.hsNorm_eq_sqrt_trace}
in \lean{Analysis/NormalizedHilbertSchmidt.lean}.
The norm is defined on actual complex matrices and agrees with
$\norm{\cdot}_{2,d}$.\\[5pt]
Definition~\ref{def:approx}: defect and uniform asymptotic triviality &
\lean{Analysis.MatrixAssignment},
\lean{Analysis.IsApproxRepresentation},
\lean{Analysis.relationDefect},
\lean{Analysis.relationDefect_le_iff}, and
\lean{Analysis.ApproximatelyTrivial}
in \lean{Analysis/ApproxRepresentation.lean}.
The defect is zero for an empty relator set. The quantifier order is
$\forall\eta>0\;\exists\delta>0\;\forall d>0\;\forall f$,
with conclusion $\norm{f(v)-I}_{2,d}<\eta$.\\[5pt]
Lemma~\ref{lem:ultraproduct}: equivalence and unitary lifting &
\lean{Analysis.approximatelyTrivial_iff_matrixQuotient_killed}
in \lean{Analysis/ApproximationFiniteCriterion.lean};
\lean{Analysis.exists_matrixQuotient_unitary_lift}
in \lean{Analysis/MatrixUnitaryLifting.lean}.
The criterion includes arbitrary target homomorphisms and varying
matrix dimensions. Every quotient unitary has unitary matrix lifts.\\[5pt]
Theorem~\ref{thm:normalization}: analytic ingredients &
\lean{Analysis.exists_matrixALT_internal_relative_commutant_and_center}
in \lean{Analysis/MatrixALTTheorem5_2.lean};
\lean{Analysis.matrixThom_internal_normalization}
in \lean{Analysis/MatrixThomInternalNormalization.lean}.
These prove internality from a spectral gap, and normalization
from compression inclusions and an internal common-commutant identity.\\[5pt]
Theorem~\ref{thm:normalization}: application to the specified $Q,H$ &
\lean{Analysis.matrixThom_groupCentralizer_normalized}
in \lean{Analysis/MatrixThomGroupNormalization.lean};
\lean{Construction.compressor_spectral_gaps} and
\lean{Construction.compressor_internality}
in \lean{Construction/UnconditionalApproximation.lean}.
The first gives the unitary-group centralizer formulation; the latter
discharge the analytic hypotheses for this pair. The general
Kazhdan/infranormal statement is not reproduced verbatim.\\
\end{longtable}

\begin{longtable}{@{}p{0.29\linewidth}p{0.68\linewidth}@{}}
\caption{The compressor pair and central obstruction.}
\label{tab:lean-groups}\\
\toprule Paper statement & Lean declarations and coverage\\ \midrule
\endfirsthead
\toprule Paper statement & Lean declarations and coverage (continued)\\ \midrule
\endhead
\bottomrule\endfoot
Proposition~\ref{prop:finite-cover}: finite presentation and property~(T) &
\lean{PrimeFiveRankThreeCover.Cover} in
\lean{Analysis/PrimeFiveRankThreeCover.lean};
\lean{Analysis.coverRootGeneratingSet_nonzero_invariant}
in \lean{Analysis/PrimeFiveCoverKazhdan.lean}.
The presentation and positive Kazhdan constant are verified for
the rank-three covers.\\[5pt]
Proposition~\ref{prop:finite-cover}: covering maps &
\lean{Compressor.coverToElementary_surjective} and
\lean{Compressor.coverToPositive_surjective}
in \lean{Groups/CompressorCoverSurjections.lean}.
These are the required instances $\Delta_6\twoheadrightarrow N$ and
$\Delta_3\twoheadrightarrow H$; the assertion for every coefficient
algebra $R$ is not separately formalized.\\[5pt]
Proposition~\ref{prop:kazhdan-pair}: $N\triangleleft Q$ and property~(T) &
\lean{Compressor.elementarySubgroup_normal}
in \lean{Groups/CompressorNormal.lean};
\lean{Compressor.elementaryKazhdanSet_nonzero_invariant} and
\lean{Compressor.positiveKazhdanSet_nonzero_invariant}
in \lean{Analysis/CompressorFiniteKazhdanSets.lean};
\lean{Compressor.hasFiniteHilbertKazhdanSet}
in \lean{Groups/CompressorKazhdan.lean}.
These give normality and property~(T) for $N,H,Q$, respectively.\\[5pt]
Proposition~\ref{prop:kazhdan-pair}: compression and infranormality &
\lean{Compressor.shear_conjugate_subgroup_le}
in \lean{Groups/CompressorCompression.lean};
\lean{Compressor.withShears_eq_top} and
\lean{Compressor.compression_generates}
in \lean{Groups/CompressorGeneration.lean}.
All six shears compress $H$ and, together with $H$, generate $Q$.\\[5pt]
Propositions~\ref{prop:Laurent-quotient} and
\ref{prop:strict-compression} &
\lean{LaurentModel.representation} and
\lean{LaurentModel.representation_ofGenerator}
in \lean{Groups/LaurentRelations.lean};
\lean{Compressor.t_conjugate_subgroup_lt} and
\lean{Compressor.h_not_mem_conjugate_positiveSubgroup}
in \lean{Groups/LaurentSeparation.lean}.
The specified homomorphism into the Laurent semidirect product
proves $tHt^{-1}<H$ and $h\notin tHt^{-1}$.\\[5pt]
Proposition~\ref{prop:double-obstruction}: specified instance &
\lean{Double.amalgamEquiv} in \lean{Groups/DoubleAmalgam.lean};
\lean{Double.obstruction_infiniteOrder}
in \lean{Construction/CentralNontrivial.lean};
\lean{Double.hom_obstruction_eq_one_of_normalizes_centralizer}
in \lean{Groups/DoubleCentralizer.lean}.
These implement the normal-form and centralizer arguments for the
specified double. The source group is not an arbitrary pair $Q,H$.\\[5pt]
Proposition~\ref{prop:central-obstruction}: specified instance &
\lean{Lambda.JElement_square} and \lean{Lambda.JElement_central}
in \lean{Groups/LambdaRelations.lean};
\lean{Lambda.JElement_ne_one}
in \lean{Construction/CentralNontrivial.lean};
\lean{Construction.lambda_matrix_J_eq_one}
in \lean{Construction/UnconditionalApproximation.lean}.
For the specified $\Lambda$, $J$ is a nontrivial central involution
and every matrix-quotient representation kills $J$.
The general proposition for arbitrary $D,w$ is not a single formal
entry theorem.\\
\end{longtable}

\subsection{Finite construction and passage to the game}

\begin{longtable}{@{}p{0.29\linewidth}p{0.68\linewidth}@{}}
\caption{Finite presentations, matrix entries, and the transfer of $J$.}
\label{tab:lean-construction}\\
\toprule Paper statement & Lean declarations and coverage\\ \midrule
\endfirsthead
\toprule Paper statement & Lean declarations and coverage (continued)\\ \midrule
\endhead
\bottomrule\endfoot
Lemma~\ref{lem:canonical-lambda}: canonical presentation and counts &
\lean{Lambda.normalizationEquiv},
\lean{Lambda.normalizationEquiv_of}, and
\lean{Lambda.normalized_cyclicallyReduced} in
\lean{Groups/LambdaNormalization.lean};
\lean{Lambda.normalized_relator_count} and
\lean{Lambda.normalized_relator_total_length}
in \lean{Groups/LambdaChecks.lean}.
Normalization preserves the group and its generators and yields
cyclically reduced words; checked
certificates give $12{,}863$ relators of total length $118{,}095$.\\[5pt]
The involution embedding \eqref{eq:first-embedding} and
Lemma~\ref{lem:collegial-input} &
\lean{Construction.lambdaToInvolution_injective}
in \lean{Construction/InvolutionFaithful.lean};
\lean{Construction.involutionPresentation_collegial}
in \lean{Construction/InvolutionGroup.lean}.
The first map is proved injective. The wheel-word checks, including
cyclic reduction, even letter counts, and minimum length, are in
\lean{Construction/Wheel.lean}.\\[5pt]
Proposition~\ref{prop:solution-embedding}: homomorphism and
preservation of $J$ &
\lean{Construction.lambdaToSigma},
\lean{Construction.lambdaToSigma_of}, and
\lean{Construction.lambdaToSigma_J}
in \lean{Construction/WheelHom.lean};
\lean{Construction.J_sigma_ne_one}
in \lean{Construction/SigmaNontrivial.lean}.
The generator formula and preservation of $J$ are proved;
$J_\Sigma\ne1$ is proved separately by diagram normalization and
wheel collapse. Full injectivity of $\Lambda\to\Sigma$, asserted
in the proposition, is not proved in the formal development. The homomorphism preserving $J$, together with the separate
proof that $J_\Sigma\ne1$, suffices for the main separation. \\[5pt]
Proposition~\ref{prop:matrix-size}: dimensions and entries &
\lean{Construction.row_count}, \lean{Construction.col_count}, and
\lean{Construction.nonzero_count}
in \lean{Construction/Counts.lean};
\lean{Construction.A_eq_one_iff},
\lean{Construction.sourceTriples_injective}, and
\lean{Construction.b_formula}
in \lean{Construction/MatrixData.lean}.
These identify all three distinct entries of every row and
$b=e_{1{,}417{,}141}$, as well as the stated counts.\\[5pt]
Proposition~\ref{prop:matrix-size}: parity and inconsistency &
\lean{Construction.wheelWord_count_even}
in \lean{Construction/Wheel.lean};
\lean{Construction.A_b_no_solution}
in \lean{Construction/MatrixData.lean}.
The former and the incidence formulas give the column-parity
calculation; there is no separate entry theorem for all column sums.
The latter proves inconsistency using the final odd wheel.\\[5pt]
Proposition~\ref{prop:uniform-solution} and
\eqref{eq:main-quantifiers} &
\lean{Construction.paper_approx_to_engineering},
\lean{Construction.engineering_approx_to_paper}, and
\lean{Construction.paperApproximation_hsNorm}
in \lean{Construction/PaperApproximation.lean}.
The presentation comparison has error factor at most four.
The final estimate uses the manuscript's literal relator convention
and gives its uniform quantifiers over all positive dimensions.\\
\end{longtable}

\begin{longtable}{@{}p{0.29\linewidth}p{0.68\linewidth}@{}}
\caption{Game-value arguments.}
\label{tab:lean-game-proofs}\\
\toprule Paper statement & Lean declarations and coverage\\ \midrule
\endfirsthead
\toprule Paper statement & Lean declarations and coverage (continued)\\ \midrule
\endhead
\bottomrule\endfoot
Proposition~\ref{thm:SV}: the implication used here &
\lean{SparseSystem.near_perfect_approximation}
in \lean{Quantum/NearPerfectApproximation.lean}.
For three-variable rows and the incidence distribution, a
sufficiently successful finite projective strategy produces an
approximate representation with $J=-I$. The converse and the
general equivalence stated in the proposition are not formalized
by this declaration.\\[5pt]
Corollary~\ref{cor:quantum-strict} and the POVM bridge &
\lean{Construction.paper_omegaQ_lt_one}
in \lean{Construction/PaperQuantumGap.lean};
\lean{Quantum.FinitePOVMStrategy.exists_projective}
in \lean{Quantum/FinitePOVMDilation.lean};
\lean{Construction.paper_omegaQPOVM_lt_one}
in \lean{Construction/PaperPOVMGap.lean}.
Finite local dilation preserves the entire correlation table and
therefore transfers the strict bound to arbitrary POVMs.\\[5pt]
Proposition~\ref{prop:commuting}: perfect commuting strategy &
\lean{Construction.paperCommutingStrategy} and
\lean{Construction.paperCommutingStrategy_perfect}
in \lean{Construction/PaperCommutingStrategy.lean}.
The formal model keeps $\ell^2(\Sigma)$ with state
$(\delta_1-\delta_{J_\Sigma})/\sqrt2$; every rejected outcome
has zero probability. This realizes the same perfect score as
the restriction to $\HH_-$ in the text. Inclusion of projective
strategies among POVM strategies gives the POVM conclusion.\\[5pt]
Proposition~\ref{prop:classical} and Remark~\ref{rem:numerical-values}:
classical value and explicit quantum gap &
\lean{Construction.paper_classical_value_corollary} and
\lean{Construction.paper_quantum_gap_corollary}
in \lean{Construction/PaperValueCorollaries.lean}.
The former proves
$\omega_c(G)=1-1/4{,}251{,}456$, including shared randomness.
The latter proves
$1-\omega_q(G)\ge 2^{-2^{50003}}$
for arbitrary finite dimensional POVM strategies.
Both are closed theorems with no additional mathematical hypotheses.\\


\end{longtable}

\subsection{Verification commands and source records}

\begin{longtable}{@{}p{0.29\linewidth}p{0.68\linewidth}@{}}
\caption{Checks to run from the repository root.}
\label{tab:lean-checks}\\
\toprule Verification target & Command or record\\ \midrule
\endfirsthead
\toprule Verification target & Command or record (continued)\\ \midrule
\endhead
\bottomrule\endfoot
Theorem~\ref{thm:existence} &
\texttt{lake -{}-fail-fast build \textbackslash}\newline
\hspace*{1em}\lean{ThomGame.Construction.NonlocalGameSeparation}
followed by
\texttt{lake env lean} \lean{scripts/CheckSeparation.lean}.
The script checks the binary-payoff statement and a version with
all four alphabets nonempty, and audits their axiom dependencies.\\[5pt]
All project modules and final statements &
\texttt{lake -{}-fail-fast build}, followed by
\texttt{lake env lean} \lean{scripts/CheckMain.lean},
\texttt{lake env lean} \lean{scripts/CheckPOVM.lean}, and
\texttt{lake env lean} \lean{scripts/CheckCorollaries.lean}.\\[5pt]
Project-wide axiom audit &
\texttt{lake env lean} \lean{scripts/AuditAxioms.lean}.
Recursively checks every imported project theorem; only
\lean{propext}, \lean{Classical.choice}, and \lean{Quot.sound}
are permitted. Additional axioms, including \lean{sorryAx},
cause failure.\\[5pt]
Pinned versions &
\lean{lean-toolchain}: \lean{leanprover/lean4:v4.35.0-rc3}.
The mathlib commit in \lean{lakefile.toml} is
\lean{16efc2c756299924184fe015f6384ad6538d42c7}.\\[5pt]
Source version and verification records &
Source release \lean{v0.1.0}, commit
\lean{714041680d5a7125d0f475d63241e433fce65713},
with dated build and axiom-audit records in the repository's
\lean{validation/} directory. The verification commands above
use the repository source and its pinned dependencies.\\
\end{longtable}
\endgroup

The build commands may reuse an existing cache. A fresh source
copy without the project's \lean{.lake/build} directory rebuilds the
project dependencies of the selected target. The guide
\lean{validation/minimal-separation-check.md} describes this scope;
the dated logs inside \lean{validation/} record earlier checks.
The commands are reproducibility instructions; the logs are records of
the runs they identify. Lean checks the formal derivations and their
axiom dependencies. The tables above separately identify the definitions,
hypotheses, and conclusions that must be compared with the paper.

\clearpage

\end{document}